\documentclass[final,11pt]{article}
\usepackage{amsfonts, color, morefloats, pslatex, a4wide}
\usepackage{amssymb, amsthm, amsmath, latexsym, cases}
\newtheorem{theorem}{Theorem}
\newtheorem{lemma}[theorem]{Lemma}

\newtheorem{corollary}[theorem]{Corollary}

\newtheorem{open}[theorem]{Open Problem}

\newtheorem{example}[theorem]{Example}

\newtheorem{conj}[theorem]{Conjecture}

\usepackage[para]{threeparttable}

\newcommand{\ord}{{\mathrm{ord}}}

\newcommand{\rank}{{\mathrm{rank}}}
\newcommand{\lcm}{{\mathrm{lcm}}}
\newcommand{\tr}{{\mathrm{Tr}}}

\newcommand{\gf}{{\mathrm{GF}}}
\newcommand{\PG}{{\mathrm{PG}}}

\newcommand{\PAut}{{\mathrm{PAut}}}
\newcommand{\MAut}{{\mathrm{MAut}}}
\newcommand{\GAut}{{\mathrm{Aut}}}

\newcommand{\TT}{{\mathtt{T}}}
\newcommand{\wt}{{\mathtt{wt}}}

\newcommand{\Z}{\mathbb{{Z}}}

\newcommand{\m}{\mathbb{M}}

\newcommand{\C}{{\mathcal{C}}}

\newcommand{\cSim}{{\mathrm{Sim}}}
\newcommand{\cHam}{{\mathrm{Ham}}}

\newcommand{\ba}{{\mathbf{a}}}

\newcommand{\bc}{{\mathbf{c}}}

\newcommand{\bu}{{\mathbf{u}}}
\newcommand{\bv}{{\mathbf{v}}}

\newcommand{\bg}{{\mathbf{g}}}
\newcommand{\be}{{\mathbf{e}}}
\newcommand{\bG}{{\mathbf{G}}}
\newcommand{\bH}{{\mathbf{H}}}

\newcommand{\bzero}{{\mathbf{0}}}
\newcommand{\bone}{{\mathbf{1}}}

\newcommand{\0}{\textbf{0}}
\newcommand{\1}{\textbf{1}}

\makeatletter

\newcommand{\Rmnum}[1]{\expandafter\@slowromancap\romannumeral #1@}
\makeatother

\begin{document}

\title{The Extended Codes of Some Linear Codes\thanks{
Z. Sun's research was supported by The National Natural Science Foundation of China under Grant Number  62002093. 
C. Ding's research was supported by The Hong Kong Research Grants Council, Proj. No. $16301123$. 
}}

\author{
Zhonghua Sun\thanks{School of Mathematics, Hefei University of Technology, Hefei, 230601, Anhui, China. Email:  sunzhonghuas@163.com}  \and 
Cunsheng Ding\thanks{Department of Computer Science
                           and Engineering, The Hong Kong University of Science and Technology,
Clear Water Bay, Kowloon, Hong Kong, China. Email: cding@ust.hk} \and Tingfang Chen \thanks{Department of Computer Science
                           and Engineering, The Hong Kong University of Science and Technology,
Clear Water Bay, Kowloon, Hong Kong, China. Email: tchenba@connect.ust.hk}
}

\maketitle

\begin{abstract} 
The classical way of extending an $[n, k, d]$ linear code $\C$ is to add an overall parity-check coordinate to each codeword of the linear code $\C$. This extended code, denoted by $\overline{\C}(-\bone)$ and called the standardly extended code of $\C$, is a linear code with parameters $[n+1, k, \bar{d}]$, where $\bar{d}=d$ or $\bar{d}=d+1$. 
This is one of the two extending techniques for linear codes in the literature.  
The standardly extended codes of some families of binary linear codes have been studied to some extent. However, not much is known about the standardly extended codes of nonbinary codes. For example, the minimum distances of the standardly extended codes of the nonbinary Hamming codes remain open for over 70 years. The first objective of this paper is to introduce the nonstandardly extended codes of a linear code and develop some general theory for this type of extended linear codes. The second objective is to study 
this type of extended codes of a number of families of linear codes, including cyclic codes and nonbinary Hamming codes.  Four families of distance-optimal or dimension-optimal linear codes are obtained with this extending technique.  The parameters of certain extended codes of many families of linear codes are settled in this paper. 

\vspace*{.3cm}
\end{abstract}


\section{Introduction}

\subsection{Linear codes, constacyclic and cyclic codes}  

Throughout this section, let $q$ be a prime power, $\gf(q)$ be the finite field with $q$ elements and $n$ be a positive integer. An $[n, k, d]$ code $\C$ over $\gf(q)$ is $k$-dimensional linear subspace of $\gf(q)^n$ with minimum distance $d$. By the parameters of a linear code, we mean its length, dimension and minimum distance. Let $A_i$ denote the total number of codewords with weight $i$ in $\C$. Then the sequence $(A_0, A_1, \ldots, A_n)$ is called the weight distribution of $\C$ and the polynomial $\sum_{i=0}^n A_i z^i$ is referred to as the weight enumerator of $\C$. The dual code of $\C$, denoted by $\C^\perp$, is defined by
$$\left\{(c_0,c_1, \ldots, c_{n-1})\in \gf(q)^n: \sum_{i=0}^{n-1} c_i b_i=0, ~\forall ~(b_0,b_1,\ldots, b_{n-1})\in \C  \right\}.$$
It is known that $\C^\perp$ is an $[n, n-k]$ code over $\gf(q)$.

An $[n, k, d]$ code over $\gf(q)$ is said to be {\it distance-optimal} (resp. {\it distance-almost-optimal}) if there is no $[n, k, d']$ code over $\gf(q)$ with $d'>d$ (resp. $d'>d+1$). An $[n, k, d]$ code over $\gf(q)$ is said to be {\it dimension-optimal} (resp. {\it dimension-almost-optimal}) if there is no $[n, k', d]$ code over $\gf(q)$ with $k' > k$ (resp. $k'>k+1$). A linear code is said to be {\it optimal} if it is distance-optimal, or dimension-optimal. By the best known parameters of $[n, k]$ linear codes over $\gf(q)$ we mean an $[n, k, d]$ code over $\gf(q)$ with the largest known $d$ reported in the tables of best linear codes known maintained at http://www.codetables.de.

Let $\lambda \in \gf(q) \backslash \{0\}$. An $[n,k,d]$ code $\C$ over $\gf(q)$ is {\it $\lambda$-constacyclic} if 
$(c_0,c_1, \ldots, c_{n-1}) \in \C$  
implies $(\lambda c_{n-1}, c_0, c_1, \ldots, c_{n-2}) \in \C$. Define
\begin{align*}
\Phi: \ \gf(q)^n &\longrightarrow 	\gf(q)[x]/\langle x^n-\lambda\rangle \\
(c_0,c_1, \ldots, c_{n-1})&\longmapsto \sum_{i=0}^{n-1} c_i x^i.
\end{align*}
It is known that every ideal of $\gf(q)[x]/\langle  x^n-\lambda \rangle$ is {\it principal} and a code $\C \subset \gf(q)^n$ is 
$\lambda$-constacyclic if and only if $\Phi(\C)$ is an ideal of $\gf(q)[x]/\langle x^n-\lambda \rangle$. Due to this, we will 
identify $\C$ with $\Phi(\C)$ for any $\lambda$-constacyclic code $\C$ of length $n$ over $\gf(q)$. Let $\C=\langle g(x) \rangle$ be a $\lambda$-constacyclic code of length $n$ over $\gf(q)$, where $g(x)$ is monic and has the smallest degree. Then $g(x)$ is called the {\it generator polynomial} and $h(x)=(x^n-\lambda)/g(x)$ is referred to as the {\it check polynomial} of $\C$. The dual code $\C^\perp$ is generated by the reciprocal polynomial of the check polynomial $h(x)$ of $\C$ and is a $\lambda^{-1}$-constacyclic code of length $n$ over $\gf(q)$ \cite{KS90}. By definition, $1$-constacyclic codes are the classical cyclic codes. 

\subsection{Automorphism groups and equivalence of linear codes} 

Two linear codes $\C_1$ and $\C_2$ are said to be  {\em permutation-equivalent\index{permutation equivalent for codes}} if there is a permutation of coordinates which sends $\C_1$ to $\C_2$. This permutation could be described by employing a \textit{permutation matrix}\index{permutation matrix}, which is a square matrix with exactly one 1 in each row and column and 0s elsewhere. The set of coordinate permutations that map a code $\C$ to itself forms a group, which is referred to as the \emph{permutation automorphism group\index{permutation automorphism group of codes}} of $\C$ and denoted by $\PAut(\C)$. 

A \emph{monomial matrix\index{monomial matrix}} over $\gf(q)$ is a square matrix having exactly one nonzero element of $\gf(q)$  in each row and column. A monomial matrix $\mathbf{M}$ can be written either in the form $\mathbf{D}\mathbf{P}$ or the form $\mathbf{P}\mathbf{D}_1$, where $\mathbf{D}$ and $\mathbf{D}_1$ are diagonal matrices and $\mathbf{P}$ is a permutation matrix. Two linear codes $\C_1$ and $\C_2$ of the same length over $\gf(q)$ are said to be \emph{scalar-equivalent\index{scalar-equivalent}} if there is an invertible diagonal matrix $\mathbf{D}$ over $\gf(q)$ such that $\C_2=\C_1\mathbf{D}:=\{\bc \mathbf{D} :~\bc \in \C_1 \}$. 

Let $\C_1$ and $\C_2$ be two linear codes of the same length over $\gf(q)$. Then $\C_1$ and $\C_2$ are said to be \emph{monomially-equivalent\index{monomially equivalent}} if there is a nomomial matrix $\mathbf{M}$ over $\gf(q)$ such that $\C_2=\C_1\mathbf{M}$. Monomial equivalence and permutation equivalence are precisely the same for binary codes. If $\C_1$ and $\C_2$ are monomially-equivalent, then they have the same weight distribution. The set of monomial matrices that map $\C$ to itself forms the group $\MAut(\C)$, which is called the \emph{monomial automorphism group\index{monomial automorphism group}} of $\C$. By definition, we have $\PAut(\C) \subseteq \MAut(\C)$. 

Two codes $\C_1$ and $\C_2$ are said to be \textit{equivalent}\index{equivalent} if there is a monomial matrix $\mathbf{M}$ and an automorphism $\gamma$ of $\gf(q)$ such that $\C_1=\C_2 \mathbf{M} \gamma$.  All three are the same if the codes are binary; monomial equivalence and equivalence are the same if the field considered has a prime number of elements. 

The \textit{automorphism group}\index{automorphism group} of $\C$, denoted by $\GAut(\C)$, is the set of maps of the form $\mathbf{M}\gamma$, where $\mathbf{M}$ is a monomial matrix and $\gamma$ is a field automorphism, that map $\C$ to itself. In the binary case, $\PAut(\C)$,  $\MAut(\C)$ and $\GAut(\C)$ are the same. If $q$ is a prime, $\MAut(\C)$ and $\GAut(\C)$ are identical. In general, we have $ \PAut(\C) \subseteq \MAut(\C) \subseteq \GAut(\C)$. 

\subsection{Extended codes $\overline{\C}(\bu)$ of linear codes}

An important topic in coding theory is the construction of new linear codes over finite fields with interesting 
parameters and properties. 
Given a linear code $\C$ over a finite field $\gf(q)$, one can obtain a new code from $\C$ in different ways. 
For example, puncturing and shortening $\C$ on a set of coordinate positions are two ways for obtaining new 
linear codes. Both the puncturing and shortening approaches are interesting and important, as every 
linear code $\C$ with dual distance at least $3$ is permutation-equivalent to a punctured one-weight irreducible 
cyclic code and a shortened code of the dual code of such  one-weight irreducible 
cyclic code \cite{HD19,LDT21}.  Another way for obtaining a new code from a given linear code $\C$ is to extend 
$\C$. There are different ways to extend a given linear code $\C$. Below is one of the two extending techniques.

Let $\bu=(u_1, u_2, \ldots, u_n) \in \gf(q)^n$ be any nonzero 
vector. Any given $[n, k, d]$ code $\C$ over $\gf(q)$ can be extended into an  $[n+1, k, \overline{d}]$ code 
$\overline{\C}(\bu)$ over $\gf(q)$, where 
\begin{eqnarray}\label{eqn-extendedcodegeneral}
\overline{\C}(\bu)=\left\{(c_1, c_2, \ldots, c_n, c_{n+1}):  (c_1,c_2, \ldots, c_n) \in \C, \, c_{n+1}=\sum_{i=1}^n u_i c_i\right\}. 
\end{eqnarray} 
By definition, $\overline{d}=d$ or $\overline{d}=d+1$. The given linear code $\C$ has  $q^n-1$ extended codes, as 
$\bu$ has $q^n-1$ choices. But many of them are not interesting. 

When $\bu \in \C^\perp$, the extended coordinate of each codeword in $\overline{\C}(\bu)$ is always 
zero and the extended code $\overline{\C}(\bu)$ is not interesting, as the extended coordinate is not 
useful for error-detection and error-correction. Such extended code is said to be \emph{trivial}.  

When $\bu$ has weight less than $n$, some coordinates of the codewords in $\C$ are not involved in the computation of the extended coordinate of all the codewords in the  extended code $\overline{\C}(\bu)$, i.e., the extended coordinate of the codewords in $\overline{\C}(\bu)$ is independent of some coordinates in the codewords in $\C$. Such an extended code $\overline{\C}(\bu)$ is said to be \emph{incomplete}. When $\bu$ has weight $n$, $\overline{\C}(\bu)$ is called a \emph{completely extended code.}

For each $[n, k, d]$ code $\C$ over $\gf(q)$, define 
$\Pi(\C)=\gf(q)^n \setminus \C^\perp.$  
Let $\bone$ and $\bzero$ denote the all-one vectors and the zero vectors over $\gf(q)$ with various lengths. Any extended code of a zero code $\{\bzero\}$ is another zero code. 
Hence, we do not consider the extended codes of the zero codes. For every nonzero code $\C$, the set $\Pi(\C)$ 
is nonempty. The extended code $\overline{\C}(-\bone)$ of certain linear codes 
$\C$ has been studied in the literature and will be called the \emph{standardly extended code}. 
For each $a \in \gf(q)^*$, the extended code $\overline{\C}(a \bone)$ is monomially-equivalent to 
$\overline{\C}(- \bone)$ and they have the same minimum distance and weight distribution. 
All extended codes $\overline{\C}(a \bone)$ with $a \neq 0$ are called the \emph{standardly extended codes}.  
All the other extended codes $\overline{\C}(\bu)$ with $\bu \neq a \bone$ for all $a \in \gf(q)^*$ are called 
\emph{nonstandardly extended codes}. Some nonstandardly extended codes $\overline{\C}(\bu)$ may be 
monomially-equivalent to a standardly extended code $\overline{\C}(a \bone)$. But the nonstandardly extended codes $\overline{\C}(\bu)$ 
have various kinds of weight distributions in general and belong to some different monomially-equivalent classes. 
In this paper, 
we consider those extended codes $\overline{\C}(\bu)$ with only $\bu \in \Pi(\C)$.   

\subsection{Some motivations of studying extended codes $\overline{\C}(\bu)$ of linear codes}

\subsubsection{First motivation: the standardly extended codes of linear codes are not well studied yet} 

There are a very small number of references and results on the  standardly extended codes $\overline{\C}(-\bone)$ of linear codes $\C$ 
over finite fields (see Section \ref{sec-sum} for a summary of earlier works). The minimum distance of the standardly extended code $\overline{\C}(-\bone)$ of most known linear codes $\C$ is open, not to mention the weight distribution of the extended code $\overline{\C}(-\bone)$.

In short, not much on the  standardly extended codes $\overline{\C}(-\bone)$ of known families of linear codes $\C$ over finite fields 
is known. Studying the standardly extended codes $\overline{\C}(-\bone)$ of known linear codes $\C$ may get new results 
on the original codes $\C$, and may obtain linear codes with new parameters or weight distributions or new applications. 
This is the first motivation of studying extended codes $\overline{\C}(\bu)$.

\subsubsection{Second motivation: nonstandardly extended codes of linear codes could be much more interesting in some senses compared with the standardly extended codes}

If $\bone \in \C^\perp$ for a given linear code $\C$, then the standardly extended code $\overline{\C}(-\bone)$ is trivial and not interesting at all. 
Even if $\overline{\C}(-\bone)$ is nontrivial and interesting in some senses, some nonstandardly extended code  $\overline{\C}(\bu)$ 
could be very much more interesting. Below is a very convincing example for elaborating this statement  above. 

\begin{example}\label{exam-111} 
Let $\lambda$ be a primitive element of $\gf(4)$ with $\lambda^2+\lambda+1=0$. It can be verified that 
$x^7+\lambda=(x+\lambda)(x^3 + \lambda^2 x + 1)(x^3 + \lambda x^2 + 1)$. Let $\C$ denote the $\lambda$-constacyclic 
code of length $7$ over $\gf(4)$ with generator polynomial $x^3 + \lambda x^2 + 1$. Then $\C$ and $\C^\perp$ have parameters $[7,4,3]$ and $[7,3,4]$, respectively. 
The standardly extended code $\overline{\C}(-\bone)$ has parameters $[8, 4, 3]$ and weight enumerator 
$ 
1 +  6z^3 + 18z^4 + 48 z^5  + 108z^6 +  42z^7 + 33 z^8. 
$  
The dual code $\overline{\C}(-\bone)^\perp$ has parameters $[8, 4, 3]$. Hence,  the standardly extended code $\overline{\C}(-\bone)$ 
is nontrivial. 

However, the nonstandardly extended code $\overline{\C}((\lambda,   1,   \lambda,   \lambda,  \lambda^2,   1, \lambda^2))$ of $\C$  has 
parameters $[8,4,4]$ and weight enumerator $1+42z^4+168z^6+45z^8$. The dual code  $\overline{\C}((\lambda,   1,   \lambda,   \lambda,  \lambda^2,   1, \lambda^2))^\perp$ has also parameters $[8,4,4]$.  This nonstandardly extended code $\overline{\C}((\lambda,   1,   \lambda,   \lambda,  \lambda^2,   1, \lambda^2))$ is much more interesting than the standardly extended code $\overline{\C}(-\bone)$ in 
the following senses. 
\begin{itemize}
\item The nonstandardly extended code $\overline{\C}((\lambda,   1,   \lambda,   \lambda,  \lambda^2,   1, \lambda^2))$ is distance-optimal  and has a larger minimum distance $4$. 
\item The nonstandardly extended code $\overline{\C}((\lambda,   1,   \lambda,   \lambda,  \lambda^2,   1, \lambda^2))$ has only 3 
nonzero weights, while the standardly extended code $\overline{\C}(-\bone)$ has 6 nonzero weights. 
\item The minimum weight codewords in  $\overline{\C}((\lambda,   1,   \lambda,   \lambda,  \lambda^2,   1, \lambda^2))$ support a 
$3$-design, while the minimum weight codewords in  $\overline{\C}(-\bone)$ do not support a $1$-design. This makes a big difference. 
\item The automorphism group of  $\overline{\C}((\lambda,   1,   \lambda,   \lambda,  \lambda^2,   1, \lambda^2))$ has size $4032$, 
while the automorphism group of $\overline{\C}(-\bone)$ has size only $12$. This makes a very big difference between these two extended codes, as the size of the automorphism group is related to the decoding complexity of the code and possibly some 
applications of the code (e.g., the construction of $t$-designs with a code). 
\end{itemize} 
\end{example}
As will be seen later, for several infinity families linear codes $\C$, the standardly extended code $\overline{\C}(-\bone)$ is trivial, but the following hold.
\begin{enumerate}
\item There is $\bu \in \Pi(\C)$ such that $\overline{\C}(\bu)$ is an MDS code (see Theorems \ref{INF-T} and \ref{INF-TT}).
\item There is $\bu \in \Pi(\C)$ such that $\overline{\C}(\bu)^\perp$ is a distance-optimal code (see Theorem \ref{INF-T}).	
\end{enumerate}
Hence, while $\overline{\C}(-\bone)$ is trivial, some nonstandardly extended codes $\overline{\C}(\bu)$ of some linear codes 
could be very interesting. 
These examples demonstrate that certain nonstandardly extended codes $\overline{\C}(\bu)$ could be much more interesting in several senses than the standardly extended code $\overline{\C}(-\bone)$ for some linear codes $\C$. But little is known about the nonstandardly extended codes $\overline{\C}(\bu)$. Notice that  the extended code $\overline{\C}(-\zeta \bone)$ of a   quadratic-residue code treated in the literature is a standardly extended code  by definition, as $\overline{\C}(-\zeta \bone)$ is monomially-equivalent to $\overline{\C}(-\bone)$. Studying the nonstandardly extended codes  $\overline{\C}(\bu)$ of some linear codes may obtain interesting codes with new parameters and properties as well as new applications, as demonstrated by Example \ref{exam-111}. This is the second and also a major motivation of studying the extended codes $\overline{\C}(\bu)$.  Hopefully, Example \ref{exam-111} and this second motivation could convince the reader that it is really interesting to study some nonstandardly extended codes $\overline{\C}(\bu)$ of some linear codes. 


\subsection{Some  earlier works on the standardly extended code $\overline{\C}(-\bone)$}\label{sec-sum} 

The standardly extended code $\overline{\C}(-\bone)$ of certain families of cyclic codes over finite fields has been studied to some extent in the literature. Affine-invariant codes are standardly extended cyclic codes and they include the extended narrow-sense primitive BCH codes \cite{BC99}, the generalized Reed-Muller codes and the Dilix codes \cite[Chapt. 6]{DingBk2}. The standardly extended code $\overline{\C}(-\bone)$ of the binary and ternary Golay code has been well investigated  \cite{HP03}. A standardly extended code $\overline{\C}(-\zeta\bone)$ of the odd-like quadratic-residue codes was studied \cite[Sec. 6.6.3]{HP03}, and its automorphism group was settled \cite{Huff95}. The extended codes $\overline{\C}(-\bone)$ of the binary Hamming and Simplex codes were analysed. The standardly extended code of the narrow-sense Reed-Solomon code and some generalized Reed-Solomon codes is known to be MDS (\cite[Sect. 11.5]{MS77}, \cite[Sect. 5.3]{HP03}). The standardly extended codes of a family of reversible MDS cyclic codes were studied in \cite{SD23}. But extended codes of most MDS codes have not been investigated. The standardly extended codes of some duadic codes are known to be self-dual, but their parameters are unknown \cite[Chapt. 6]{HP03}. Some extended Goppa codes were treated in \cite{BM73}, \cite[Chapt. 12]{MS77}. The standardly extended codes of the narrow-sense primitive BCH codes were studied in \cite{DZ16} and the standardly extended codes of several families of cyclic codes were investigated in \cite{AD21,DDT,Ding18,DingBk2,Pless86}. 

\subsection{Another type of extended codes in the literature} 

\subsubsection{The description of this type of extended codes}

For a given $[n, k,d]$ linear code $\C$ over $\gf(q)$, another extending method goes as follows. Choose an $\ell \times n$ matrix $\bG$ whose row vectors linearly span $\C$ and $\bg \in \gf(q)^\ell$. 
Then construct an $\ell \times (n+1)$ matrix 
\begin{eqnarray}\label{eqn-tildeGv}
\underline{\bG}(\bG, \bg)=(\bG~\bg^T),
\end{eqnarray} 
where $\bg^T$ denotes the transpose of $\bg$. Let $\underline{\C}(\bG, \bg)$ be the linear code linearly spanned by the row vectors of the matrix $\underline{\bG}(\bG,\bg)$. By definition, $\ell \ge k$ and the extended code $\underline{\C}(\bG, \bg)$ has parameters $[n+1, \underline{k},\underline{d}]$, where $\underline{k} \in \{k, k+1\}$ and $\underline{d} \in \{d, d+1\}$.  It is easy to see that the following hold:
\begin{itemize}
\item The exact values of $\underline{k}$ and $\underline{d}$ depend on both $\bG$ and $\bg$. 
\item The weight distribution of the extended code $\underline{\C}(\bG, \bg)$ depends on both $\bG$ and $\bg$.   
\item The extended code $\underline{\C}(\bG, \bg)$ has dual distance $1$ if $\bg$ is the zero vector.  
\end{itemize} 
There are a huge number of choices of such matrix $\bG$ for the given $\C$ and many choices of a nonzero  vector $\bg$. Hence, by this extending technique a given $[n, k, d]$ linear code $\C$ over $\gf(q)$ could be extended into a huge number of linear codes $\underline{\C}(\bG, \bg)$. 

\subsubsection{Differences and connections between the two types of extended codes}

For a given $[n, k,d]$ linear code $\C$ over $\gf(q)$, the two kinds of extended codes $\overline{\C}(\bu)$ and $\underline{\C}(\bG, \bg)$ are different but related in certain cases.  While the extending technique $\overline{\C}(\bu)$ does not change the dimension of $\C$, the extending technique $\underline{\C}(\bG, \bg)$ may increase the dimension of $\C$ by one in certain cases.  This is one of the major differences between 
the two types of extending techniques.  
We have the following lemmas whose proofs are straightforward and omitted.

\begin{lemma}\label{lem-nov121} 
For each $\overline{\C}(\bu)$ with any $1 \times n$ vector $\bu \in \gf(q)^n$, let $\bG$ be any generator matrix of $\C$ and $\bg^T=\bG\bu^{T}$. Then $\overline{\C}(\bu)=\underline{\C}(\bG, \bg)$. 
\end{lemma} 

Note that there are many pairs of $\bG$ and $\bg$ such that $\overline{\C}(\bu)=\underline{\C}(\bG, \bg)$ according to Lemma \ref{lem-nov121}. Hence, any result on an extended code $\overline{\C}(\bu)$ can be transformed into the same result on many extended codes $\underline{\C}(\bG, \bg)$.  

\begin{lemma}\label{lem-nov131}
For a given extended code $\underline{\C}(\bG, \bg)$ with a $1 \times \ell$ vector $\bg \in \gf(q)^\ell$, there is a $1 \times n$ vector in $\gf(q)^n$ such that $\overline{\C}(\bu)= \underline{\C}(\bG, \bg)$ if and only if  $\bG\bu^T=\bg^T$. 
\end{lemma} 

Lemma \ref{lem-nov131} says that a result on an extended code $\underline{\C}(\bG, \bg)$ may not be transformed into the same result on any extended code $\overline{\C}(\bu)$. Hence, the two extending techniques are different.  

\begin{lemma}\label{lem-nov122} 
For a given extended code  $\underline{\C}(\bG, \bg)$ with $\bG$ being a $k \times n$ generator matrix of $\C$ and $\bg$ being a $1 \times k$ vector over $\gf(q)$,  there are $q^{n-k}$ vectors $\bu$ in $\gf(q)^n$ such that $\overline{\C}(\bu)=\underline{\C}(\bG, \bg)$, as $\rank(\bG)=\rank((\bG~\bg^T))=k \leq n$. These $\bu$ are the solutions to the system of equations $\bG\bu^T=\bg^T$. 
\end{lemma} 

Lemma \ref{lem-nov122} documents a special case in which a result on an extended code $\underline{\C}(\bG, \bg)$ can be transformed to the same result on certain extended codes $\overline{\C}(\bu)$. 

According to Lemma \ref{lem-nov131}, the extended coordinate in an extended code $\underline{\C}(\bG, \bg)$ may not depend on other coordinates linearly in certain cases. Naturally, coding theorists are more interested in the extended codes $\overline{\C}(\bu)$, as the extended coordinate in $\overline{\C}(\bu)$ depends on other coordinates linearly according to the specific relation defined by $\bu$ and is redundant. The purpose of adding the extended coordinate in $\overline{\C}(\bu)$ is for error-detection and error-correction. In particular, coding theorists are mostly interested in these extended codes $\overline{\C}(\bu)$ with $\bu$ having full Hamming weight $n$. This explains why the classical extended code $\overline{\C}(-\bone)$ for certain families of linear codes was investigated by coding theorists in the literature.

\subsubsection{Some earlier references on the type of extended codes $\underline{\C}(\bG, \bg)$}

Let $\C$ be an $[n, k, d]$ linear code over $\gf(q)$. If there is an $[n+1, k, d+1]$ linear code $\C'$ such that $\C$ is a punctured code of $\C'$, then $\C$ is said to be \emph{extendable} and $\C'$ is called an \emph{extension} of $\C$. A code $\C$ is \emph{doubly extendable} if there exists an extension of $\C$ which is also extendable. Multiple extendability is similarly defined. 

The extendability of certain linear codes $\C$ with respect to this extending technique $\underline{\C}(\bG, \bg)$ and parameters of certain extended linear codes $\underline{\C}(\bG, \bg)$  were investigated in \cite{CM09,Hill99,HLL94,Hirsch98,KM18,KSMM22,KTM16,Kohnert,LR06,LRS16,Muruta01,Muruta03,Muruta04,Muruta051,Muruta052,Muruta08,Muruta11,MO06,MO07,MO09,MTK08,MY12,Roth,SR,Simonis,YM10}. The results about some extended linear codes $\underline{\C}(\bG, \bg)$ presented in these references may not be transformed into similar results about any extended code $\overline{\C}(\bu)$ according to Lemma \ref{lem-nov131}. Notice that this paper will not investigate this type of extended codes $\underline{\C}(\bG, \bg)$.

\subsection{The objectives of this paper} 

This paper is mainly concerned with the following open problem:
\begin{open}\label{prob-1120} 
What are the parameters of an extended code $\overline{\C}(\bu)$ for a given linear code  $\C$ and vector $\bu$? 
\end{open} 

The objectives of this paper are the following:
\begin{enumerate}
\item Present several fundamental results about the extended codes $\overline{\C}(\bu)$ of linear codes over finite fields.
\item Study the parameters of the extended codes $\overline{\C}(\bu)$ of some infinite families of linear codes over finite fields.
\item Analyse the parameters of the extended codes $\overline{\C}(\bu)$ of nonbinary Hamming codes.
\item Construct several infinite families of linear codes with optimal parameters with the extending technique $\overline{\C}(\bu)$.	
\item Solve Open Problem \ref{prob-1120}  for certain families of linear codes $\C$ and certain vectors $\bu$.
\end{enumerate}

\subsection{The organisation of this paper}

The rest of this paper is organized as follows. In Section \ref{Sec:2}, we present some auxiliary results. In Section \ref{Sec:3}, we prove some general results for the extended codes. In Section \ref{Sec:4}, we study the parameters of the dual of the extended codes. In Section \ref{sec-5555}, we analyse the parameters of the extended codes of several infinite families of linear codes over finite fields. We present two families of near MDS codes. In Section \ref{nsec::5}, we analyse the parameters of the standardly extended codes of several families of cyclic codes. We obtain several infinite families of distance-optimal codes and distance-almost-optimal codes. In Section \ref{Sec:6}, we analyse the parameters of the extended codes of nonbinary Hamming codes. In Section \ref{Sec:7}, we conclude this paper and make some concluding remarks. 

\section{Some auxiliary results}\label{Sec:2}

Throughout this paper, for a linear code $\C$, we use $\dim(\C)$ and $d(\C)$ to denote its dimension and minimum distance, respectively.

\subsection{Cyclotomic cosets}

Let $n$ be a positive integer with $\gcd(q, n)=1$, $r$ be a positive divisor of $q-1$, and let $\lambda$ be an element of $\gf(q)$ with order $r$. To deal with $\lambda$-constacyclic codes of length $n$ over $\gf(q)$, we have to study the factorization of $x^n-\lambda$ over $\gf(q)$. To this end, we need to introduce $q$-cyclotomic cosets modulo $rn$. 

Let $\Z_{rn}=\{0,1,2,\cdots,rn-1\}$ be the ring of integers modulo $rn$. For any $i \in \Z_{rn}$, the \emph{$q$-cyclotomic coset of $i$ modulo $rn$} is defined by 
\[C^{(q,rn)}_i=\{iq^j \bmod {rn}: \ 0 \leq j\leq \ell_i-1 \} \subseteq \Z_{rn}, \]
where $\ell_i$ is the smallest positive integer such that $i \equiv i q^{\ell_i} \pmod{rn}$, and is the \textit{size} of the $q$-cyclotomic coset $C^{(q,rn)}_i$. The smallest integer in $C^{(q, rn)}_i$ is called the \textit{coset leader} of $C^{(q, rn)}_i$. Let $\Gamma_{(q, rn)}$ be the set of all the coset leaders, then
 $$\bigcup_{i \in  \Gamma_{(q, rn)} } C_i^{(q,rn)}=\Z_{rn}.$$

Let $m=\ord_{r n}(q)$. It is easily seen that there is a primitive element $\alpha$ of $\gf(q^m)$ such that $\beta=\alpha^{(q^m-1)/rn}$ and $\beta^n=\lambda$. Then $\beta$ is a primitive $rn$-th root of unity in $\gf(q^m)$. 
 The \textit{minimal polynomial} $\m_{\beta^i}(x)$ of $\beta^i$ over $\gf(q)$ is the monic polynomial of the smallest degree over $\gf(q)$ with $\beta^i$ as a zero. We have 
 $$
\m_{\beta^i}(x)=\prod_{j \in C_i^{(q,rn)}} (x-\beta^j) \in \gf(q)[x], 
$$ 
which is irreducible over $\gf(q)$. 
Define $\Gamma_{(q, rn,r)}^{(1)}=\{i: i \in \Gamma_{(q,rn)}, \, i \equiv 1 \pmod{r} \}$. Then
 $$x^{n}-\lambda=\prod_{i \in  \Gamma_{(q,rn,r)}^{(1)}} \m_{\beta^i}(x).$$

\subsection{The BCH bound for constacyclic codes}

The following lemma documents the BCH bound for constacyclic codes over finite fields, which is a generalization of the BCH bound of cyclic codes. 

\begin{lemma}\label{lem-BCHbound}\cite[Lemma 4]{KS90}
Let $\gcd(n, q)=1$ and $\beta$ be a primitive $rn$-th root of unity such that $\beta^n=\lambda$. Let $\mathcal{C}=\langle g(x)\rangle$ be a $\lambda$-constacyclic code of length $n$ over $\gf(q)$. If there are integers $b$ and $\delta$ with $2\leq \delta \leq n$ such that
$$g(\beta^{b})=g(\beta^{b+r})=\cdots=g(\beta^{b+r(\delta-2)})=0,$$
then $d(\C)\geq \delta$.
\end{lemma}

\subsection{The trace representation of constacyclic codes}
 
 The trace representation of $\lambda$-constacyclic codes is documented below (see \cite{DY10}, \cite{SR2018}, \cite[Theorem 1]{SZW20}). 

\begin{lemma}\label{lem-01} 
Let $\gcd(n, q)=1$ and $\beta$ be a primitive $rn$-th root of unity such that $\beta^n=\lambda$. Let $\C$ be the $\lambda$-constacyclic code of length $n$ over $\gf(q)$ with check polynomial $\prod_{j=1}^s \m_{\beta^{i_j}}(x)$, where $C_{i_{a}}^{(q, rn)} \cap C_{i_{b}}^{(q, r n)}=\emptyset$ for $a\neq b$. Then $\C$ has the trace representation 
$$\left \{ \left(\sum_{j=1}^s{\rm Tr}_{q^{m_j}/q}(a_j\beta^{-ti_j})\right)_{t=0}^{n-1} \,:\,a_j\in {\rm GF}(q^{m_j}),\,1\leq j\leq s \right \},$$ 
where $m_j=|C_{i_j}^{(q, rn)}|$ and $\tr_{q^m/q}$ denotes the trace function from $\gf(q^m)$ to $\gf(q)$.
\end{lemma}

Lemma \ref{lem-01} is very useful in determining the parameters and weight distributions of some constacyclic codes. We will make use of this lemma later in this paper.

\subsection{The sphere packing bound and Pless power moments}

We recall the sphere packing bound for linear codes in the following lemma.

\begin{lemma}
{\rm (Sphere Packing Bound \cite{HP03})} Let $\C$ be an $[n,k,d]$ code over $\gf(q)$. Then 
$$\sum_{i=0}^{\lfloor (d-1)/2 \rfloor}\binom{n}{i}(q-1)^i\leq q^{n-k},$$
where $\lfloor \cdot \rfloor$ is the floor function.	
\end{lemma}

\subsection{The maximal length of MDS codes, AMDS codes and NMDS codes}

Let $\C$ be an $[n, k, d]$ code over $\gf(q)$. The parameters of $\C$ satisfy the Singleton bound: $d\leq n-k+1$. If $d=n-k+1$, the code $\C$ is called {\it maximum distance separable} (MDS for short). If $d=n-k$, the code $\C$ is called {\it almost maximum distance separable} (AMDS for short). The code $\C$ is said to be {\it near maximum distance separable} (NMDS for short) 
if both $\C$ and $\C^\perp$ are AMDS. The following results will be needed later in this paper.

\begin{lemma} \label{lem:9}
{\rm \cite{Roth}} Let $\C$ be an $[n, k]$ MDS code over $\gf(q)$. If $q$ is odd and $2\leq k<(\sqrt{q}+13)/4$, then $n\leq q+1$.
\end{lemma}

\begin{lemma}\label{AMDS} 
{\rm (\cite{Bose1952,Mario})} Let $q>2$ be a prime power. Let $\C$ be an $[n, n-4]$ AMDS code over $\gf(q)$. Then $n\leq q^2+1$.
\end{lemma}

\section{Some general theory about the extended codes $\overline{\C}(\bu)$}\label{Sec:3}

\begin{theorem}\label{thm-general111} 
Let $\C \subset \gf(q)^n$ be a linear code and $\bu \in \gf(q)^n \setminus \{\bzero\}$.   
Then the following statements are equivalent. 
\begin{enumerate}
\item The extended code $\overline{\C}(\bu)$ is trivial. 
\item $\bu$ belongs to $\C^\perp$.  
\item  The vector $(0, \ldots, 0, 1) \in \overline{\C}(\bu)^\perp$. 
\end{enumerate} 
\end{theorem} 

\begin{proof}
The desired conclusion follows directly from the definition of the trivially extended codes.  
\end{proof} 

The following theorem shows the importance of the extending technique. 

\begin{theorem}\label{thm-general112} 
Every $[n, k, d]$ code $\C \subset \gf(q)^n$ with $d>1$ is permutation-equivalent to 
the extended code $\overline{\C'}(\bu)$ of a linear code $\C' \subset \gf(q)^{n-1}$ 
for some $\bu \in \gf(q)^{n-1}$.  
\end{theorem} 

\begin{proof}
Let $\C$ be an $[n, k, d]$ code over $\gf(q)$ with $d>1$. Then $\C \neq \gf(q)^n$ and $\C^\perp$ is not the 
zero code. Let $\bv=(v_1, v_2, \ldots, v_n)$ be any nonzero codeword in $\C^\perp$. Then $\wt(\bv) \geq 1$. 
Then there is an integer $i$ with $1 \leq i \leq n$ and $v_i \neq 0$. 
Define $\bv'=v_i^{-1} \bv$. Then the $i$-th coordinate $v'_i$ of $\bv'$ is $1$ and $\bv' \in \C^\perp$. 

If $i=n$, let $P$ be the identity permutation on the coordinate set $\{1,2, \cdots, n \}$. If $i \neq n$, let $P$ be the transposition $(i, n)$ on the coordinate set $\{1,2, \cdots, n\}$, which is a special permutation of $\{1,2, \cdots, n\}$. For each $(c_{1}, c_{2}, \ldots, c_n) \in \gf(q)^{n}$, define 
$$ 
P((c_1, c_2, \ldots, c_n))=(c_{P(1)}, c_{P(2)}, \ldots, c_{P(n)}). 
$$
Then $P$ is a permutation of $\gf(q)^n$. 
Define 
$$ 
\bu=-P(\bv')^{\{n\}} \in \gf(q)^{n-1}, 
$$ 
where $P(\bv')^{\{n\}}$  is the vector obtained by puncturing the last coordinate in $P(\bv')$. 
Let   
$$
\C'=P(\C)^{\{n\}}, 
$$   
which is the linear code obtained by puncturing the linear code $P(\C)$ on the $n$-th coordinate. Note that $d(P(\C))=d(\C)=d>1$. According to \cite[Theorem 1.5.1]{HP03}, we have 
$$ 
\dim(\C')=\dim(P(\C))=\dim(\C)=k.  
$$ 
It is well known that 
$$
d(\C')=d \mbox { or } d-1. 
$$
It is easily seen that 
$$ 
P(\C)=\overline{\C'}(\bu). 
$$ 
Consequently, 
$$
\C=P(\overline{\C'}(\bu)).
$$ 
This completes the proof. 
\end{proof}

The following is a corollary of Theorem  \ref{thm-general111} and its proof.  It says that almost all cyclic codes are extended linear codes. 

\begin{corollary} 
Every $[n, k, d]$ cyclic code $\C \subset \gf(q)^n$ with $d>1$ is  
the extended code $\overline{\C^{\{n\}}}(\bu)$ for some $\bu \in \gf(q)^{n-1}$, 
where $\C^{\{n\}}$ denotes the punctured 
code of $\C$ at the last coordinate position and $\bu$ is obtained by puncturing the last coordinate in a nonzero codeword of $\C^\perp$.
\end{corollary} 

The following is another corollary of Theorem  \ref{thm-general112} and its proof.  

\begin{corollary}\label{thm-general113} 
Let $\C$ be an $[n, k, d]$ code over $\gf(q)$ with $d>1$. Then 
there are a positive integer $t$, a sequence $(P_1, P_2,\ldots, P_t)$  
of permutations, a sequence $(\bu_1, \bu_2, \ldots, \bu_t) \in \gf(q)^{n-t} \times \gf(q)^{n-t+1} \times \cdots \times \gf(q)^{n-1}$ 
of vectors,  
 and an $[n-t, k, 1]$ code $\C'$ over $\gf(q)^n$ such that 
$$
\C=P_t(\overline{P_{t-1}(\overline{ \cdots P_1(\overline{\C'}(\bu_1)) \cdots}(\bu_{t-1})) }(\bu_{t})), 
$$
where each $P_i$ is either the identity permutation or a transposition of the set $\{1,2, \cdots, n-t+i\}$. 
\end{corollary} 

Corollary \ref{thm-general113} says that every $[n, k, d]$ code over $\gf(q)$ with $d>1$ can be 
obtained by performing a sequence of extensions of a special code $\C'$ with minimum distance $1$ plus some necessary coordinate transpositions in the middle of the extensions.   

\section{Some basic results of the extended code $\overline{\C}(\bu)$}\label{Sec:4} 

\subsection{A connection among some extended codes $\overline{\C}(\bu)$ }  

Let $\C$ be a given $[n, k, d]$ code over $\gf(q)$ and let $\Pi(\C)=\gf(q)^n \setminus \C^\perp$. Suppose that 
$\bu \in \Pi(\C)$ has weight $n$. Define $P_\bu$ to be the diagonal matrix with $-\bu$ as the diagonal vector and $ \C'=\C P_\bu$. Then $\C'$ and $\C$ are scalar-equivalent and have thus the same parameters and weight distribution. By definition, 
\begin{eqnarray}
\overline{\C'}(-\bone)=\overline{\C}(\bu). 
\end{eqnarray} 
This means that the nonstandardly extended code  $\overline{\C}(\bu)$ of $\C$ is the same as the standardly extended code $\overline{\C'}(-\bone)$ of the related code $\C'$ defined above. As will be seen later, a nonstandardly extended code may be very different from the standardly extended code of a linear code. On the other hand, Example \ref{exam-111} and some results later show that the standardly extended codes $\overline{\C_1} (-\bone)$ and  $\overline{\C_2} (-\bone)$ could be very different even if $\C_1$ and $\C_2$ are scalar-equivalent or monomial-equivalent. 

\subsection{The algebraic structure of the extended code $\overline{\C}(\bu)$}

Let $\C$ be a given $[n, k, d]$ code over $\gf(q)$ and $\bu \in \gf(q)^n$. If $\C$ has generator matrix $\bG$ and parity check matrix $\bH$. It follows from (\ref{eqn-extendedcodegeneral}) that the generator and parity check matrices for $\overline{\C}(\bu)$ are $\begin{pmatrix}
		\bG & \bG \bu^\TT
	\end{pmatrix}$
	and 
	$$\begin{pmatrix}
  \bH	& \0^\TT\\
 \bu   & -1
\end{pmatrix},
  $$
  where $\bu^\TT$ denotes the transpose of $\bu$.

The following theorem describes a relation between the codes $\overline{\C}(\bu)^\perp$ and $\C^\perp$.  

\begin{theorem}\label{thm:1}
Let $\C\subset \gf(q)^n$ and $\bu \in \gf(q)^n$. Then 
$$\overline{\C}(\bu)^{\bot}=\left\{ (\bc-a \bu, a):~\bc\in \C^{\bot},~a\in \gf(q) \right\} $$	
and $
d(\overline{\C}(\bu)^{\bot}) \leq d(\C^{\bot})$.
\end{theorem}

\begin{proof}
Let $\bu=(u_1,u_2,\ldots,u_n)\in \gf(q)^n$. Let $\bc=(c_1,c_2,\ldots,c_n)\in \C^{\bot}$ and $a\in \gf(q)$. For any $(b_1, b_2, \ldots, b_{n+1})\in \overline{\C}(\bu)$, it is easily verified that 
\begin{align*}
&\sum_{i=1}^{n}b_i(c_i-a u_i)+ab_{n+1}\\
=&~\sum_{i=1}^{n}b_i(c_i-a u_i)+a\sum_{i=1}^n u_i b_i\\
=&~\sum_{i=1}^{n}b_ic_i=0,
\end{align*}
where the last equlality follows from the fact that $(b_1,b_2,\ldots,b_n)\in \C$. Therefore, 
$$\left\{ (\bc-a \bu, a):~\bc\in \C^{\bot},~a\in \gf(q) \right\}\subseteq \overline{\C}(\bu)^{\bot}.$$ 
On the other hand, it is easily checked that 
\begin{align*}
&\left| \left \{ (\bc-a \bu, a):~\bc\in \C^{\bot},~a\in \gf(q)\right\}\right|\\
=&~q |\C^{\bot}|=|\overline{\C}(\bu)^{\bot}|.	
\end{align*}
Therefore, 
$$\left\{ (\bc-a \bu, a):~\bc\in \C^{\bot},~a\in \gf(q) \right\}=\overline{\C}(\bu)^{\bot}.$$ 
This completes the proof.
\end{proof}

The vector $\bu \in \gf(q)^n$ and the rows of a generator matrix of $\C$ span a linear subspace denoted by $\widetilde{\C}(\bu)$, which is called the {\it augmented code} of $\C$ by $\bu$. If $\bu \in \C$, then $\widetilde{\C}(\bu)=\C$. Otherwise, $\C$ is a subcode of $\widetilde{\C}(\bu)$, and $\dim(\widetilde{\C}(\bu))=\dim(\C)+1$ and $d(\widetilde{\C}(\bu))\leq d(\C)$. The augmented code $\widetilde{\C}(\bone)$ of certain linear codes has been studied in the literature and will be called the {\it standardly augmented code}. The minimum distance of the dual code $\overline{\C}(\bu)^{\bot}$ has the following results.

\begin{theorem}\label{dual-distance}
 Let $\bu \in \Pi(\C)$, then 
\begin{align*}
d(\overline{\C}(\bu)^{\bot})	=\begin{cases}
	d(\widetilde{\C^\perp}(\bu))+1~&{\rm if}~d(\widetilde{\C^{\bot}}(\bu))<d(\C^{\bot}),\\
	d(\widetilde{\C^\perp}(\bu))~&{\rm if}~d(\widetilde{\C^{\bot}}(\bu))=d(\C^{\bot}).
\end{cases}
\end{align*}
\end{theorem}

\begin{proof} For any $\0 \neq \overline{\bc}\in \overline{\C}(\bu)^{\bot}$, by Theorem \ref{thm:1}, there are $\bc \in \C^{\bot}$ and $a\in \gf(q)$ such that $\overline{\bc}=(\bc-a \bu, a)$. It is easily seen that
\begin{align*}
\wt (\overline{\bc})=\begin{cases}
	\wt(\bc)~&{\rm if}~a=0,\\
	\wt(\bc-a\bu)+1~&{\rm if}~a\neq 0.
\end{cases}
\end{align*}
If $a=0$, we have $\bc\neq \0$. Then $\wt(\overline{\bc} )\geq d(\C^{\bot})$. Since $\bu \notin \C^\perp$, we deduce that $\bc-a\bu\neq \0$. Then $\wt(\overline{\bc})\geq d( \widetilde{\C^{\bot}}(\bu))+1$ for $a\neq 0$. Therefore, 
$$\wt(\overline{\bc} )\geq \min\left\{d(\C^{\bot}), d( \widetilde{\C^{\bot}}(\bu))+1 \right\}.$$ 
It follows that 
\begin{equation*}
d(\overline{\C}(\bu)^\bot)\geq \min\left\{d(\C^{\bot}), d( \widetilde{\C^{\bot}}(\bu))+1 \right\}.	
\end{equation*} 
The rest of the proof is then divided into the following two cases.
\begin{enumerate}
\item[Case 1.] Suppose $d(\widetilde{\C^{\bot}}(\bu))<d(\C^{\bot})$. Then $d(\overline{\C}(\bu)^\bot)\geq d( \widetilde{\C^{\bot}}(\bu))+1$, and there are $\bc\in \C^{\bot}$ and $a\in \gf(q)^*$ such that $\wt(\bc-a \bu)=d(\widetilde{\C^{\bot}}(\bu))$. It is easily seen that 
$$\wt((\bc-a \bu, a))=d(\widetilde{\C^{\bot}}(\bu))+1.$$
It follows that $d(\overline{\C}(\bu)^{\bot})\leq d(\widetilde{\C^{\bot}}(\bu))+1$. The desired result follows.

\item[Case 2.] Suppose $d(\widetilde{\C^{\bot}}(\bu))=d(\C^{\bot})$. Then $d(\overline{\C}(\bu)^\bot)\geq d( \C^{\bot})$, and there is $\bc\in \C^{\bot}$ such that $\wt(\bc)=d(\widetilde{\C^{\bot}}(\bu))$. It is easily seen that $\wt((\bc, 0))=d(\widetilde{\C^{\bot}}(\bu))$. It follows that 
$$d(\overline{\C}(\bu)^{\perp})\leq d(\widetilde{\C^{\perp}}(\bu)).$$ 
The desired result follows. 
\end{enumerate}
This completes the proof. 	
\end{proof}

\begin{theorem}\label{lem:2}
Let $\C\subseteq \gf(q)^n$ be a linear code with $d(\C^{\bot})\geq 3$ and let $\bu \in \Pi(\C)$. Then the following hold.
\begin{enumerate}
\item $d(\overline{\C}(\bu)^{\perp} )\geq 2$.	
\item $d(\overline{\C}(\bu)^{\perp} )=2$ if and only if there are $a\in \gf(q)^*$ and $\be_j$ such that $\bu+a \be_j\in \C^{\perp}$, where $\be_j$ is the vector of $\gf(q)^n$ such that its $j$-th coordinator is $1$ and its other coordinates are $0$.
\item Let $\overline{A}_2^{\perp}$ denote the total number of codewords with weight $2$ in $\overline{\C}(\bu)^{\bot}$, then $\overline{A}_2^{\perp}\in \{0,q-1\}$.
\end{enumerate}
\end{theorem}

\begin{proof}
By Theorem \ref{thm:1}, $d(\overline{\C}(\bu)^\perp )= 1$ if and only if there are $\bc\in \C^{\perp}$ and $a\in \gf(q)$ such that $\wt((\bc-a \bu, a ))=1$. Suppose on the contrary that $d(\overline{\C}(\bu)^\perp )= 1$. Then there are $\bc\in \C^{\perp}$ and $a\in \gf(q)$ such that $\wt((\bc-a \bu, a ))=1$.
\begin{itemize}
\item If $a=0$, then $\wt((\bc-a \bu, a ))=\wt(\bc)\geq 3$, a contradiction.	
\item  If $a\neq 0$, then $\wt((\bc-a \bu, a ))=1$ if and only if $\bc=a\bu$. It follows that $\bu \in \C^{\perp}$, a contradiction.
\end{itemize}
Therefore, $d(\overline{\C}(\bu)^{\perp} )\geq 2$.

By Theorem \ref{thm:1}, $d(\overline{\C}(\bu)^{\perp} )= 2$ if and only if there are $\bc\in \C^{\perp}$ and $a\in \gf(q)$ such that $$\wt((\bc-a \bu, a ))=2.$$ 
\begin{itemize}
	\item If $a=0$, then $\wt((\bc-a \bu, a ))=\wt(\bc)\geq 3$, a contradiction. 
	\item If $a\neq 0$, $\wt((\bc-a \bu, a ))=2$ if and only if  $\wt(\bc-a \bu)=1$. It follows that there are $b\in \gf(q)^*$ and $1\leq j\leq n$ such that $\bc-a \bu=b {\bf e}_j$, i.e., $a\bu +b{\bf e}_j\in \C^{\perp}$. Note that  $a\bu+b \be_j \in \C^{\perp}$ if and only if $\bu+a^{-1}b \be_j\in \C^{\perp}$. 
\end{itemize}
The second desired result follows.

By Result 2, we deduce 
$$A_2^\perp=|\{(a, b, j): a \bu+b \be_j\in \C^\perp,~a,~b\in \gf(q)^*, 1\leq j\leq n \}|.$$ 
Now suppose there are $a$ and $j$ such that $\bu+a\be_j \in \C^{\perp}$. If there are $a_1,a_2\in \gf(q)^*$ and $1\leq j_1,j_2\leq n$ such that $\bu+a_i\be_{j_i}\in \C^{\perp}$, then $a_1\be_{j_1}-a_2\be_{j_2}\in \C^{\perp}$. Note that $\wt(a_1 \be_{j_1}-a_2 \be_{j_2})\leq 2$. Since $d(\C^{\bot})\geq 3$, it follows that $a_1=a_2$ and $j_1=j_2$. Therefore, there is a unique pair $(a, j)$ such that $\bu+a\be_{j}\in \C^{\perp}$. It follows that $\overline{A}_2^{\perp}=q-1$. The third desired result follows. 
This completes the proof. 
\end{proof}

\begin{corollary}\label{CO-13}
Let $q>2$ be a prime power. Let $\C^{\bot}$ be an $[n, k, d^\perp]$ code over $\gf(q)$ and $d^{\perp}\geq 3$. Let $\bu\in \Pi(\C)$. The following hold.
\begin{enumerate}
\item If $\C^{\perp}$ is monomially-equivalent to the Hamming code, then $d(\overline{\C}(\bu)^{\perp})=2$ and $\overline{A}_2^{\perp}=q-1$.
\item Otherwise, there always exists a vector $\bu\notin \C^{\perp}$ such that $d(\overline{\C}(\bu)^{\perp})=2$ and there also exists a vector $\bu\notin \C^{\perp}$ such that $d(\overline{\C}(\bu)^{\perp})\geq 3$.  
\end{enumerate}	
\end{corollary}

\begin{proof}
If $\C^{\perp}$ is a perfect single error-correcting code, then any $\bu\in \gf(q)^n$ must be contained in the sphere of radius $1$ about a codeword of $\C^{\perp}$. It follows from Theorem \ref{lem:2} that $d(\overline{\C}(\bu)^{\perp})=2$ and $\overline{A}_2^{\perp}=q-1$. According to \cite[Theorem 1.12.3]{HP03}, $\C^{\perp}$ is monomially-equivalent to the Hamming code. The first desired result follows.

By assumption, $d(\C^\perp)\geq 3$. If $\C^{\bot}$ is not a perfect single error-correcting code, by the sphere packing bound, we get that $1+n(q-1)<q^{n-k}$.
Therefore, there exists a vector $ \bu\not \in \C^{\perp}$ that is not contained in the sphere of radius $1$ of any codeword of $\C^{\perp}$. On the other hand, it is clear that there exists a vector $ \bu \not\in \C^{\perp}$ that is contained in the sphere of radius $1$ of a codeword of $\C^{\perp}$. The second desired result follows. 
This completes the proof. 
\end{proof}

If $\C$ is an $[n, k, d]$ code over $\gf(q)$. It is clear that $\overline{\C}(\bu)^{\bot}$ has length $n+1$ and dimension $n+1-k$. By Theorem \ref{thm:1}, $d(\overline{\C}(\bu)^\perp)\leq d(\C^{\perp})$. In the following Theorem \ref{INF-T} we will prove that for some linear codes $\C$ there exists $\bu \in  \Pi(\C)$ such that $d(\overline{\C}(\bu))=d+1$ and $d(\overline{\C}(\bu)^{\bot})=d(\C^{\bot})$. It will be also proved later for some linear codes $\C$ there exists $\bu \in \Pi(\C)$ such that $d(\overline{\C}(\bu)^\bot)=3 \leq d(\C^\perp)$ and $\overline{\C}(\bu)^\perp$ is a distance-optimal code. 

\section{The extended codes of several families of linear codes}\label{sec-5555}

\subsection{A nonstandardly extended code $\overline{\C}(\bu)$ could be very interesting, while the standardly extended code 
$\overline{\C}(-\bone)$ is trivial}\label{sec-Nov191}
 
\begin{theorem}\label{INF-T}
Let $q>2$ be a prime power.	Let $\alpha$ be a primitive element of $\gf(q)$. Let $\C$ be the cyclic code of length $q-1$ over $\gf(q)$ with generator polynomial 
$$\sum_{i=0}^{d-2}(x-\alpha^i),$$
where $2\leq d\leq q-2$. Let $\bu=(1, \alpha^{d-1}, \alpha^{2(d-1)}, \ldots, \alpha^{(q-2)(d-1)})$, then $\overline{\C}(\bu)$ is a $[q, q-d, d+1]$ MDS code over $\gf(q)$ and $\overline{\C}(\bu)^{\perp}$ is a $[q, d, q-d+1]$ MDS code over $\gf(q)$.   
\end{theorem}

\begin{proof} 
It is well known that $\C$ is a $[q-1, q-d, d]$ MDS code over $\gf(q)$ and a cyclic Reed-Solomon code. It is easily verified that 
$$\begin{pmatrix}
	1 & 1& 1&\cdots &1&0\\
	1 & \alpha & \alpha^2 &\cdots &\alpha^{q-2}&0\\
	\vdots & \vdots & \vdots & \cdots & \vdots & \vdots \\
	1 & \alpha^{d-2} &\alpha^{(d-2)2} &\cdots &\alpha^{(d-2)(q-2)}&0\\
	1 & \alpha^{d-1} &\alpha^{(d-1)2} &\cdots &\alpha^{(d-1)(q-2)}&-1
\end{pmatrix}$$	
is a check matrix of $\overline{\C}(\bu)$. It then follows that $\overline{\C}(\bu)$ is a $[q, q-d, d+1]$ MDS code over $\gf(q)$. Consequently, $\overline{\C}(\bu)^{\perp}$ is a $[q, d, q-d+1]$ MDS code over $\gf(q)$. This completes the proof.  
\end{proof}

Let $\C$ be the cyclic code in Theorem \ref{INF-T}. It is easily seen that $\overline{\C}(-\bone)$ is trivial. However, the nonstandardly extended code $\overline{\C}(\bu)$ has minimum distance $d(\C)+1$, and $\overline{\C}(\bu)^{\perp}$ has minimum distance $d(\C^\perp)$.  The purpose of presenting Theorem \ref{INF-T} here is to show that a nonstandardly extended code $\overline{\C}(\bu)$ could be very interesting, even if the standardly extended code $\overline{\C}(-\bone)$ is trivial.  

\subsection{Multiply extended codes could be very interesting}\label{sec-nov24} 

Let $q>3$ be a prime power. Let $\alpha$ be a primitive element of $\gf(q)$. Define
\begin{equation}
\bG_\alpha:=\begin{pmatrix}
0& 1& \alpha^2 & \cdots & \alpha^{2(q-2)}\\ 
0& 1& \alpha & \cdots & \alpha^{q-2}\\
1& 1& 1& \cdots & 1\\
\end{pmatrix}.	
\end{equation}
We use $\C_\alpha$ to denote the linear code over $\gf(q)$ spanned by the rows of the matrix $\bG_\alpha$. It is easily checked that $\C_\alpha$ is a $[q, 3, q-2]$ MDS code over $\gf(q)$, and $\C_\alpha^\bot$ is a $[q, q-3, 4]$ MDS code over $\gf(q)$. The parameters of some extended codes of $\C_\alpha$ are documented in the following theorem.

\begin{theorem}\label{INF-TT}
Let notation be the same as before. The following hold.
\begin{enumerate}
\item The standardly extended code $\overline{\C_{\alpha}}(-\bone)$ is trivial.
\item Let $\bu^{(1)}=-(0,1,\alpha^{-2},\ldots, \alpha^{-2(q-2)})$. The extended code $\overline{\C_{\alpha}}(\bu^{(1)})$ has parameters $[q+1,3, q-1]$ and $\overline{\C_{\alpha}}(\bu^{(1)})^\perp$ has parameters $[q+1, q-2, 4]$. 
\end{enumerate}	
\end{theorem}

\begin{proof}
1. It is easily verified that $\mathbf{G}_\alpha \bone^\TT=(0,0,0)^\TT$. It follows that $\bone \in \C_\alpha^\perp$. The desired 
first result follows.

2. Clearly, $\sum_{i=0}^{q-2} \alpha^{s i}=0$ for each $s\in \{1,2,\cdots, q-2\}$. Therefore, 
\begin{align*}
	\bG_\alpha (\bu^{(1)})^\TT=(1,0,0)^\TT. 
\end{align*} 
It follows that  
\begin{equation*}
\bG_\alpha^{(1)}:=\begin{pmatrix}
0& 1& \alpha^2 & \cdots & \alpha^{2(q-2)}&1\\ 
0& 1& \alpha & \cdots & \alpha^{q-2}&0\\
1& 1& 1& \cdots & 1&0\\
\end{pmatrix}
\end{equation*}
is a generator matrix of $\overline{\C_\alpha}(\bu^{(1)})$.

Recall a set of $q+1$ points of $\PG(2, \gf(q))$ is said to be a {\it conic} if they are zeros of a nondegenerate quadratic form in three variables \cite[Section 1.11]{Hirsch98}. The set of column vectors of the matrix $\bG_\alpha^{(1)}$ is a conic in $\PG(2, \gf(q))$. Therefore, $\overline{\C_\alpha}(\bu^{(1)})$ is called a  {\it conic code}. The conic code has parameters $[q+1, 3, q-1]$ and its dual has parameters $[q+1, q-2, 4]$ (see \cite[Chapter 12]{DingBk2}). 
 This completes the proof.
\end{proof}

The parameters of some extended codes of the conic code $\overline{\C_\alpha}(\bu^{(1)})$ are documented in the next theorem. 

\begin{theorem}\label{ConicCode1}
	Let notation be the same as before. The following hold.
	\begin{enumerate}
	\item The standardly extended code $\overline{\overline{\C_\alpha}(\bu^{(1)})}(-\bone)$ has parameters $[q+2, 3, q-1]$ and $\overline{\overline{\C_{\alpha}}(\bu^{(1)})}(-\bone)^{\bot}$ has parameters $[q+2, q-1, 2]$.	
	\item Let $\bu^{(2)}=-(0,1,\alpha^{-1}, \ldots, \alpha^{-(q-2)}, 0)$, then we have the following results.
	\begin{itemize}
	\item[2.1.] If $q\geq 5$ is an odd prime power, the extended code $\overline{\overline{\C_\alpha}(\bu^{(1)})}(\bu^{(2)})$ has parameters $[q+2, 3, q-1]$, and $\overline{\overline{\C_{\alpha}}(\bu^{(1)})}(\bu^{(2)})^{\bot}$ has parameters $[q+2, q-1,3]$. 
	\item[2.2.] If $q\geq 4$ is a power of $2$, the extended code $\overline{\overline{\C_\alpha}(\bu^{(1)})}(\bu^{(2)})$ has parameters $[q+2, 3, q]$, and $\overline{\overline{\C_{\alpha}}(\bu^{(1)})}(\bu^{(2)})^{\bot}$ has parameters $[q+2, q-1,4]$. 
	\end{itemize}
	\end{enumerate}
\end{theorem}

\begin{proof}
	1. It is easily checked that $\mathbf{G}_{\alpha}^{(1)} \bone^\TT=(1,0,0)^\TT$. By definition, the matrix
\begin{align*}
\begin{pmatrix}
0& 1& \alpha^2 & \cdots & \alpha^{2(q-2)}&1& -1\\ 
0& 1& \alpha & \cdots & \alpha^{q-2}&0& 0\\
1& 1& 1& \cdots & 1&0& 0\\
\end{pmatrix}
\end{align*}
is a generator matrix of $\overline{\overline{\C_{\alpha}}(\bu^{(1)})}(-\bone)$. It follows that $d(\overline{\overline{\C_{\alpha}}(\bu^{(1)})}(-\bone)^\perp)=2$. It is easily seen that $d(\overline{\overline{\C_{\alpha}}(\bu^{(1)})}(-\bone))=q-1$ or $q$. If $d(\overline{\overline{\C_{\alpha}}(\bu^{(1)})}(-\bone))=q$, then $\overline{\overline{\C_{\alpha}}(\bu^{(1)})}(-\bone)$ is a $[q+2, 3, q]$ MDS code over $\gf(q)$. It then follows that $\overline{\overline{\C_{\alpha}}(\bu^{(1)})}(-\bone)^\perp$ is a $[q+2, q-1, 4]$ MDS code over $\gf(q)$, a contradiction. The desired first result follows.

2. It is easily verified that $\bG_\alpha^{(1)}(\bu^{(2)})^\TT=(0,1,0)^\TT$. It follows that 
\begin{align*}
\bG_\alpha^{(2)}:= \begin{pmatrix}
0& 1& \alpha^2 & \cdots & \alpha^{2(q-2)}&1& 0\\ 
0& 1& \alpha & \cdots & \alpha^{q-2}&0& 1\\
1& 1& 1& \cdots & 1&0& 0\\
\end{pmatrix}
\end{align*}
is a generator matrix of $\overline{\overline{\C_{\alpha}}(\bu^{(1)})}(\bu^{(2)})$. It is easily checked that any two column vectors of $\bG_\alpha^{(2)}$ are not proportional. It follows that $d(\overline{\overline{\C_{\alpha}}(\bu^{(1)})}(\bu^{(2)})^\perp)\geq 3$. The rest of the proof is divided into the following cases.
\begin{enumerate}
\item[Case 1.] Suppose $q$ is odd. For each $0\leq i<\frac{q-1}2$, let $j=\frac{q-1}2+i$. It is easy to see that 
\begin{align*}
\begin{vmatrix}
\alpha^{2i}& \alpha^{2j} & 0\\	
\alpha^{i} & \alpha^{j}  & 1\\
      1    &      1	     & 0\\
\end{vmatrix}=-(\alpha^{2i}-\alpha^{2j})=0.
\end{align*}
It follows there are $c_1,c_2\in \gf(q)$ such that $$(0,1,0)=c_1(\alpha^{2i},\alpha^{i},1)+c_2(\alpha^{2j},\alpha^{j},1).$$ 
Therefore, $d(\overline{\overline{\C_{\alpha}}(\bu^{(1)})}(\bu^{(2)})^\perp)\leq 3$. Similar to Result 1, we obtain
$$d(\overline{\overline{\C_{\alpha}}(\bu^{(1)})}(\bu^{(2)}))=q-1.$$ 
The desired result follows.   
\item[Case 2.] Suppose $q$ is even. 	The set of column vectors of the matrix $\bG_{\alpha}^{(2)}$ is a hyperoval in $\PG(2, \gf(q))$. According to \cite[Chapter 12]{DingBk2}, $\overline{\overline{\C_\alpha}(\bu^{(1)})}(\bu^{(2)})$ is a $[q+2, 3, q]$ MDS code over $\gf(q)$, and $\overline{\overline{\C_\alpha}(\bu^{(1)})}(\bu^{(2)})^\perp$ is a $[q+2, q-1, 4]$ MDS code over $\gf(q)$. The desired result follows.  
\end{enumerate}
This completes the proof.  
\end{proof}

When $q\geq 5$ is an odd prime power, $\overline{\overline{\C_\alpha}(\bu^{(1)})}(-\bone)$ and $\overline{\overline{\C_\alpha}(\bu^{(1)})}(\bu^{(2)})$ have the same parameters, but $\overline{\overline{\C_\alpha}(\bu^{(1)})}(\bu^{(2)})$ and $\overline{\overline{\C_\alpha}(\bu^{(1)})}(-\bone)$ are not equivalent. This is because the nonstandardly extended code $\overline{\overline{\C_\alpha}(\bu^{(1)})}(\bu^{(2)})$ is NMDS, while $\overline{\overline{\C_\alpha}(\bu^{(1)})}(-\bone)$ 
is not NMDS due to the fact that 
 $$d(\overline{\overline{\C_\alpha}(\bu^{(1)})}(-\bone)^\perp)=2.$$ 
 When $q\geq 4$ is a power of $2$, we have $$d(\overline{\overline{\C_\alpha}(\bu^{(1)})}(\bu^{(2)}))=d(\overline{\C_\alpha}(\bu^{(1)}))+1,$$
and $\overline{\overline{\C_\alpha}(\bu^{(1)})}(\bu^{(2)})$ is an MDS code over $\gf(q)$. However, $$d(\overline{\overline{\C_\alpha}(\bu^{(1)})}(-\bone))=d(\overline{\C_\alpha}(\bu^{(1)})).$$
These results show that the doubly extended code $\overline{\overline{\C_\alpha}(\bu^{(1)})}(\bu^{(2)})$ is more interesting 
than the doubly extended code  $\overline{\overline{\C_\alpha}(\bu^{(1)})}(-\bone)$.

The following theorem shows that the triply extended code $\overline{\overline{\overline{\C_\alpha}(\bu^{(1)})}(\bu^{(2)})}(-\bone)$  is a NMDS code over $\gf(q)$.

\begin{theorem}\label{NNMDS::17}
 The triply extended code $\overline{\overline{\overline{\C_\alpha}(\bu^{(1)})}(\bu^{(2)})}(-\bone)$ has parameters $[q+3, 3, q]$ and $$\overline{\overline{\overline{\C_\alpha}(\bu^{(1)})}(\bu^{(2)})}(-\bone)^\perp$$ has parameters $[q+3, q, 3]$. Furthermore, if $q$ is odd, $\overline{\overline{\overline{\C_\alpha}(\bu^{(1)})}(\bu^{(2)})}(-\bone)$ has weight enumerator 
$$ 1+q(q-1)z^{q}+\frac{q^3-2q^2+7q-6}{2}z^{q+1}+(2q^2-5q+3)z^{q+2}+\frac{(q-2)(q-1)^2}2z^{q+3}.$$
If $q$ is even, $\overline{\overline{\overline{\C_\alpha}(\bu^{(1)})}(\bu^{(2)})}(-\bone)$ has weight enumerator 
$$ 1+\frac{q^2+q-2}{2} z^{q}+\frac{q^3+q^2-2q}{2}z^{q+1}+\frac{q^2-q}{2}z^{q+2}+\frac{q^3-3q^2+2q}{2}z^{q+3}.$$
\end{theorem}

\begin{proof}
	It is easy to check that $\mathbf{G}_\alpha^{(2)} \bone^\TT =(1, 1, 0)^\TT$.
	It follows that 
\begin{align*}
\mathbf{G}_\alpha^{(3)}:=\begin{pmatrix}
0& 1& \alpha^2 & \cdots & \alpha^{2(q-2)}&1& 0& -1\\ 
0& 1& \alpha & \cdots & \alpha^{q-2}&0& 1& -1\\
1& 1& 1& \cdots & 1&0& 0& 0\\
\end{pmatrix}
\end{align*}
is a generator matrix of $\overline{\overline{\overline{\C_\alpha}(\bu^{(1)})}(\bu^{(2)})}(-\bone)$. It is easily verified that any two column vectors of $\mathbf{G}_\alpha^{(3)}$ are not proportional. Hence, $d(\overline{\overline{\overline{\C_\alpha}(\bu^{(1)})}(\bu^{(2)})}(-\bone)^\perp)\geq 3$. It is clear that 
$$-(1,1,0)=-(1,0,0)-(0,1,0),$$ 
then $d(\overline{\overline{\overline{\C_\alpha}(\bu^{(1)})}(\bu^{(2)})}(-\bone)^\perp)\leq 3$. Therefore, $d(\overline{\overline{\overline{\C_\alpha}(\bu^{(1)})}(\bu^{(2)})}(-\bone)^\perp)= 3$. The rest of the proof is divided into the following cases.
\begin{enumerate}
\item[Case 1.] Let $q$ be odd. For any $(a, b, c)\in \gf(q)^3$, we have 
$$\bc=(a, b, c) \mathbf{G}_\alpha^{(3)} =(c, a+b+c, \ldots, a \alpha^{2(q-2)}+b \alpha^{q-2}+c, a, b, -(a+b)).$$
It follows that 
\begin{equation}\label{NMDS-4}
\wt(\bc)=\wt((a, b, -a-b))+|\{x\in \gf(q): \ a x^2+b x+c \neq 0 \}|.	
\end{equation}
If $a=b=0$, it follows from (\ref{NMDS-4}) that 
\begin{align*}
\wt(\bc)&=|\{x\in \gf(q): \ c \neq 0 \}|\\
&=	 \begin{cases}
	 q ~&{\rm if}~c\neq 0,\\
	 0~&{\rm if}~c=0 	
	 \end{cases}\\
&=\begin{cases}
	 q ~&{\rm with}~q-1~{\rm times},\\
	 0~&{\rm with}~1~{\rm times}.	
	 \end{cases}	 
\end{align*}
If $a=0$ and $b\neq 0$, it follows from (\ref{NMDS-4}) that 
\begin{align*}
\wt(\bc)&=2+|\{x\in \gf(q): \ b x+c \neq 0 \}|\\
&=q+1~{\rm with}~(q-1)q~{\rm times}.
\end{align*}
If $a\neq 0$, it follows from (\ref{NMDS-4}) that 
\begin{equation}\label{NMDS-2}
\wt(\bc)=\wt((b,-a-b))+N(a, b, c), 	
\end{equation}
where 
\begin{align*}
N(a, b, c)&=1+|\{x\in \gf(q): \ a x^2+b x+c\neq 0\}|\\
&=1+|\{x\in \gf(q): \  ( x+2^{-1}a^{-1} b)^2 \neq  4^{-1} a^{-2}(b^2-4ac)\}|\\
&=\begin{cases}
	q~&{\rm if}~b^2=4ac,\\
	q-1~& {\rm if}~b^2-4ac\neq 0~{\rm is~the~square},\\
	q+1~&{\rm otherwise}.
\end{cases}
\end{align*}
Below we continue our proof by considering the following cases:
\begin{enumerate}
\item If $a\neq 0$ and $b=0$, it follows from (\ref{NMDS-2}) that 
\begin{align*}
\wt(\bc)&=1+N(a,b,c)\\
&=\begin{cases}
	q+1~&{\rm if}~c=0,\\
	q~& {\rm if}~-ac\neq 0~{\rm is~the~square},\\
	q+2~&{\rm otherwise} 
\end{cases}	\\
&=\begin{cases}
	q+1~&{\rm with}~(q-1)~{\rm times},\\
	q~& {\rm with}~(q-1)(\frac{q-1}2)~{\rm times},\\
	q+2~&{\rm with}~(q-1)(\frac{q-1}2)~{\rm times}.
	\end{cases}
\end{align*}
\item If $a\neq 0$ and $b=-a$, it follows from (\ref{NMDS-2}) that
\begin{align*}
\wt(\bc)&=1+N(a,b,c)\\
&=\begin{cases}
	q+1~&{\rm if}~c=4^{-1}a,\\
	q~& {\rm if}~a^2-4ac\neq 0~{\rm is~the~square},\\
	q+2~&{\rm otherwise} 
\end{cases}	\\
&=\begin{cases}
	q+1~&{\rm with}~q-1~{\rm times},\\
	q~& {\rm with}~(q-1)(\frac{q-1}2)~{\rm times},\\
	q+2~&{\rm with}~(q-1)(\frac{q-1}2)~{\rm times}.
\end{cases}	
\end{align*}
\item If $b\neq 0$ and $b\neq -a$, it follows from (\ref{NMDS-2}) that
\begin{align*}
\wt(\bc)&=2+N(a,b,c)\\
&=\begin{cases}
	q+2~&{\rm if}~b^2=4ac,\\
	q+1~& {\rm if}~b^2-4ac\neq 0~{\rm is~the~square},\\
	q+3~&{\rm otherwise} 
\end{cases}	\\
&=\begin{cases}
	q+2~&{\rm with}~(q-1)(q-2)~{\rm times},\\
	q+1~& {\rm with}~(q-1)(q-2)(\frac{q-1}2)~{\rm times},\\
	q+3~&{\rm with}~(q-1)(q-2)(\frac{q-1}2)~{\rm times}.
\end{cases}	
\end{align*}
\end{enumerate}
Summarizing the discussions above, we deduce that 
$$d(\overline{\overline{\overline{\C_\alpha}(\bu^{(1)})}(\bu^{(2)})}(-\bone))=q$$
and the desired weight enumerator holds.  

\item[Case 2.] Let $q$ be even. According to \cite[Section 12.2]{DingBk2}, the code $\overline{\overline{\C_\alpha}(\bu^{(1)})}(\bu^{(2)})$ is a projective two-weight code with weight enumerator
$$1+\frac{(q+2)(q^2-1)}{2}z^{q}+\frac{q(q-1)^2}{2}z^{q+2}.$$
Since $\overline{\overline{\C_\alpha}(\bu^{(1)})}(\bu^{(2)})$ has only nonzero weights $q$ and $q+2$, the nonzero weights of 
$$\overline{\overline{\overline{\C_\alpha}(\bu^{(1)})}(\bu^{(2)})}(-\bone)$$
must belong to the set $\{q, q+1, q+2, q+3\}$. Let $\overline{A}_i$ denote the total number of codewords with Hamming weight $i$ in $\overline{\overline{\overline{\C_\alpha}(\bu^{(1)})}(\bu^{(2)})}(-\bone)$, then
\begin{equation}\label{EQQ:7}
	\overline{A}_{q}+\overline{A}_{q+1}= \frac{(q+2)(q^2-1)}{2}
\end{equation}
and
\begin{equation}\label{EQQ:8}
\overline{A}_{q+2}+\overline{A}_{q+3}= \frac{q(q-1)^2}{2}.	
\end{equation}
Note that $d(\overline{\overline{\overline{\C_\alpha}(\bu^{(1)})}(\bu^{(2)})}(-\bone)^\perp)=3$, by the second and third Pless power moments (see \cite[Section 7.3]{HP03}), we get that
\begin{equation}\label{EQQ:9}
q\overline{A}_{q}+(q+1)\overline{A}_{q+1}+(q+2)\overline{A}_{q+2}+(q+3)\overline{A}_{q+3}=	q^2(q-1)(q+3)		
\end{equation}
and
\begin{equation}\label{EQQ:10}
q^2\overline{A}_{q}+(q+1)^2\overline{A}_{q+1}+(q+2)^2\overline{A}_{q+2}+(q+3)^2\overline{A}_{q+3}=	q(q-1)(q+3)[q-1)(q+3)+1].
\end{equation}
It is clear that 
\begin{align*}
&q\overline{A}_{q}+(q+1)\overline{A}_{q+1}+(q+2)\overline{A}_{q+2}+(q+3)\overline{A}_{q+3}\\
=~&q(\overline{A}_{q}+\overline{A}_{q+1})+(q+2)(\overline{A}_{q+2}+\overline{A}_{q+3})+	\overline{A}_{q+1}+\overline{A}_{q+3}.
\end{align*}
Combining (\ref{EQQ:7}), (\ref{EQQ:8}) and (\ref{EQQ:9}), we obtain that 
\begin{equation}\label{EQQ:11}
\overline{A}_{q+1}+\overline{A}_{q+3}= q^2(q-1).	
\end{equation}
Notice that
\begin{align*}
&q^2\overline{A}_{q}+(q+1)^2\overline{A}_{q+1}+(q+2)^2\overline{A}_{q+2}+(q+3)^2\overline{A}_{q+3}\\
=~&q^2(\overline{A}_{q}+\overline{A}_{q+1})+(q+2)^2(\overline{A}_{q+2}+\overline{A}_{q+3})+(2q+1)(\overline{A}_{q+1}+\overline{A}_{q+3})+4\overline{A}_{q+3}.
\end{align*}
Combining (\ref{EQQ:7}), (\ref{EQQ:8}), (\ref{EQQ:10}) and (\ref{EQQ:11}), we obtain that
$\overline{A}_{q+3}=\frac{q^3-3q^2+2q}{2}$. The rest of desired conclusions directly follow from Equations (\ref{EQQ:7}), (\ref{EQQ:8}) and (\ref{EQQ:11}).   
\end{enumerate}
This completes the proof.
\end{proof}

When $q\geq 5$ is odd, the doubly extended code $\overline{\overline{\C_\alpha}(\bu^{(1)})}(\bu^{(2)})$ is a NMDS code over $\gf(q)$.  The triply extended code $\overline{\overline{\overline{\C_\alpha}(\bu^{(1)})}(\bu^{(2)})}(-\bone)$ is still a NMDS code over $\gf(q)$. In this case, 
$$d(\overline{\overline{\overline{\C_\alpha}(\bu^{(1)})}(\bu^{(2)})}(-\bone))=d(\overline{\overline{\C_\alpha}(\bu^{(1)})}(\bu^{(2)}))+1.$$
This shows that the standardly extended code of a NMDS code could still be a NMDS code.  
When $q\geq 4$ is a power of $2$, the code $\overline{\overline{\overline{\C_\alpha}(\bu^{(1)})}(\bu^{(2)})}(-\bone)$ was studied by Wang and Heng \cite{WH20}. Our contribution is to show that a NMDS code can be obtained through an appropriate triple extension of the MDS code $\C_\alpha$. This is a partial demonstration of Corollary \ref{thm-general113}.   

\section{The standardly extended codes of several families of cyclic codes}\label{nsec::5}

In this section, we will study the standardly extended codes of several families of cyclic codes, and present several families of linear codes with better parameters.   

\begin{theorem}\label{Cyclic-optimal}
Let $q>2$ be a prime power. Let $n=\lambda(\frac{q^m-1}{q-1})$, where $m\geq 2$, $\lambda$ divides $q-1$ and $\gcd(m\lambda, q-1)>1$. Let $\alpha\in \gf(q^m)$ be a primitive $n$-th root of unity. Let $\C$ be the cyclic code of length $n$ over $\gf(q)$ with generator polynomial $\m_{\alpha}(x)$, where $\m_{\alpha}(x)$ denotes the minimal polynomial of $\alpha$ in $\gf(q)$. Then the following hold.
\begin{enumerate}
\item The extended code $\overline{\C}(-\bone)$ has parameters $[n+1, n-m, 3]$ and is dimension-optimal.
\item If $\lambda \geq 2$, the extended code $\overline{\C}(-\bone)$ is also distance-optimal.
\end{enumerate} 
\end{theorem}

\begin{proof}
It is easily verified that $\C$ has length $n$ and dimension $n-m$. Consequently, $\overline{\C}(-\bone)$ has length $n+1$ and dimension $n-m$. Since $$\gcd(n, q-1)=\gcd(m\lambda, q-1)>1,$$ 
we deduce that $x^{\frac{n}{\gcd(n, q-1)}}- \alpha^{\frac{n}{\gcd(n, q-1)}} \in \C$. Therefore, $d(\C)=2$. It is clear that $d(\overline{\C}(-\bone))=2$ or $3$. If $d(\overline{\C}(-\bone))=2$, then $c (x^j-x^i) \in \C$ for some $c\in \gf(q)^*$ and $0\leq i<j\leq n-1$. It follows that $\alpha^{j-i}=1$. Consequently, $n\mid (j-i)$, which contradicts the fact that $0\leq i<j\leq n-1$. Therefore, $d(\overline{\C}(-\bone))=3$. 

Suppose there exists an $[n+1, k\geq n-m+1, 3]$ code over $\gf(q)$. By the sphere packing bound, we get that 
$$1+(n+1)(q-1)\leq q^{n+1-k}\leq q^{m}.$$
It follows that $n+1\leq \frac{q^m-1}{q-1}$, a contradiction. Therefore, $\overline{\C}(-\bone)$ is a dimension-optimal code. 
 
When $\lambda\geq 2$, we have $n\geq 2(\frac{q^m-1}{q-1})$. Suppose there exists an $[n+1, n-m, d\geq 4]$ code over $\gf(q)$. By \cite[Theorem 1.5.1]{HP03}, there exists an $[n, n-m, d\geq 3]$ code over $\gf(q)$. By the sphere packing bound, we get that $1+n(q-1)\leq q^{m}$. It follows that $n\leq \frac{q^m-1}{q-1}$, a contradiction. Therefore, $\overline{\C}(-\bone)$ is a distance-optimal code.  This completes the proof. 
\end{proof} 

Note that the original code $\C$ in Theorem \ref{Cyclic-optimal} has minimum distance $2$ and is not interesting as it 
cannot correct one error. However, the extended code $\overline{\C}(-\bone)$ can correct one error and is both 
distance-optimal and dimension-optimal when 
$\lambda \geq 2$.  This demonstrates the importance of the extending technique again.

\begin{example}
We have the following examples of the code of Theorem \ref{Cyclic-optimal}.
\begin{enumerate}
\item If $(q,m, \lambda)=(3, 4, 1)$, the code $\overline{\C}(-\bone)$ has parameters $[41,36,3]$ and is dimension-optimal. 
\item If $(q, m, \lambda)=(5, 4, 1)$, the code $\overline{\C}(-\bone)$ has parameters $[157,152,3]$ and is dimension-optimal. 	
\end{enumerate}	
\end{example}

\begin{theorem}\label{MIN4:21}
Let $q$ be a prime power and $m\geq 2$ be even. Let $n$ be a divisor of $q^m-1$ and $n>q^{m/2}+1$. Let $\alpha\in \gf(q^m)$ be a primitive $n$-th root of unity and $\m_{\alpha^i}(x)$ denote the minimal polynomial of $\alpha^i$ over $\gf(q)$. Let $\C$ be the cyclic code of length $n$ over $\gf(q)$ with generator polynomial $\m_{\alpha}(x) \m_{\alpha^{ q^{m/2}+1}}(x)$. Then $d(\C)\geq 3$ and $d(\overline{\C}(-\bone))\geq 4$. Further, $d(\overline{\C}(-\bone))=4$ provided that $d(\C)=3$.
\end{theorem}

\begin{proof}
It is clear that $d(\C)\geq 2$. If $d(\C)=2$, then there are integer $i$ with $1\leq i\leq n-1$ and $c\in \gf(q)^*$ such that $x^i-c\in \C$. Consequently, 
\begin{align*}
c=\alpha^i
=\alpha^{(q^{m/2}+1)i} 
=c^{q^{m/2}+1} 
=c^2.
\end{align*}
It follows that $c=1$. Consequently, $n\mid i$, a contradiction. Therefore, $d(\C)\geq 3$.

It is clear that $d(\overline{\C})= d(\C)$ or $d(\C)+1$. Since $d(\C)\geq 3$, we get $d(\overline{\C}(-\bone))=3$ if and only if there are three pairwise distinct integers $j_1, j_2, j_3$ with $0\leq j_i\leq n-1$ and $c_1,c_2,c_3\in \gf(q)^*$ such that 
\begin{align*}
\begin{cases}
c_1+c_2+c_3=0, \\
c_1\alpha^{j_1}+c_2\alpha^{j_2}+c_3\alpha^{j_3}=0, \\
c_1\alpha^{ej_1}+c_2\alpha^{ej_2}+c_3\alpha^{ej_3}=0,
\end{cases}
\end{align*}
where $e=q^{m/2}+1$. It then follows that $c(x):=c_1 x^{j_1}+c_2x^{j_2}+c_3x^{j_3}$ is a codeword of the cyclic code $\C'$ of length $n$ over $\gf(q)$ with generator polynomial 
$$g(x):=(x-1)\m_{\alpha}(x) \m_{\alpha^e}(x).$$
Note that $g( \alpha^{i+q^{m/2}\cdot j})=0$ for any $i, j\in \{0,1\}$. By the Hartmann-Tzeng bound (see \cite[Theorem 4.5.6]{HP03}), we deduce that $d(\C')\geq 4$, a contradiction. Therefore, $d(\overline{\C}(-\bone))\geq 4$. In particular, if $d(\C)=3$, we have 
$$d(\overline{\C}(-\bone))\leq d(\C)+1=4.$$ 
Therefore, $d(\overline{\C}(-\bone))=4$. This completes the proof.  
\end{proof}

\begin{corollary}\label{OPM4:1}
Let $q$ be a prime power and $m\geq 2$ be even. Let $n=\frac{q^m-1}{e}$, where $e$ is a divisor of $q^m-1$ and $e<q^{m/4}$. Let $\C$ be the cyclic code of length $n$ over $\gf(q)$ in Theorem \ref{MIN4:21}. Then $\overline{\C}(-\bone)$ has parameters $\left[n+1, n-\frac{3m}2, 4\right]$, and is a distance-optimal code provided that 
\begin{itemize}
	\item $q=2$, $m \geq 6$ and $e<q^{\frac{m-4}4}$, or
	\item $q=3$ and $e^2< 2\cdot 3^{\frac{m}2-1}$, or
	\item $q\geq 4$.
\end{itemize}	
\end{corollary}

\begin{proof}
Since $e< q^{m/4}$, we deduce that $n>q^{m/2}+1$. It follows that $|C_1^{(3,n)}|=m$. Consequently, $\deg(\m_{\alpha}(x) )=m$. Suppose $|C_{q^{m/2}+1}^{(3,n)}|=s$, then $s\mid \frac{m}2$ and 
\begin{equation}\label{CCEQ:9}
	(q^{\frac{m}2}+1)(q^s-1)\equiv 0\pmod{n}.
\end{equation}
It follows from (\ref{CCEQ:9}) that $e (q^s-1)\equiv 0 \pmod{q^{m/2}-1}$. If $s\neq \frac{m}2$, we have $s\leq \frac{m}4$. Thereby,
\begin{align*}
	e(q^s-1) < q^{m/4} (q^{m/4}-1) 
	               <q^{m/2}-1,
\end{align*}
a contradiction. Therefore, $\deg(\m_{\alpha^{q^{m/2}+1}}(x) )=s=\frac{m}2$. Consequently, 
$$\dim(\overline{\C}(-\bone))=\dim(\C)=n-\frac{3m}2.$$

By Theorem \ref{MIN4:21}, $d(\overline{\C}(-\bone))\geq 4$. Suppose there is an $\left[n+1, n-\frac{3m}2, d\geq 5\right]$ code over $\gf(q)$. By the sphere packing bound, we get  
\begin{align*}
1+\binom{n+1}{1} (q-1)+\binom{n+1}{2} (q-1)^2\leq q^{\frac{3m}2+1}.	
\end{align*}
It follows that 
\begin{align*}
\frac{-(2q^{\frac{3m}2+1}-2q)e^2+(q^2-1)(q^m-1)e+(q-1)^2(q^{m}-1)^2}{2e^2}\leq 0.
\end{align*}
It then follows that
\begin{align*}
N:=(q-1)^2(q^{m}-1)^2-(2q^{\frac{3m}2+1}-2q)e^2<0.
\end{align*}
 
If $q=2$ and $e<q^{\frac{m-4}4}$, we have 
\begin{align*}
N &=(2^{\frac{m-4}2}-e^2)2^{\frac{3m}2+2}-2^{m+1}+4e^2+1\\
&> 2^{\frac{3m}2+2}-2^{m+1}\\
&> 0,
\end{align*}
a contradiction. If $q=3$ and $e^2< 2\cdot 3^{\frac{m}2-1}$, we have 
\begin{align*}
N &=(4\cdot 3^{\frac{m}2}-6e^2)3^{\frac{3m}2}-8\cdot 3^{m}+6e^2+4\\
&> 0,
\end{align*}
a contradiction. If $q\geq 4$, since $e<q^{m/4}$, we deduce that
\begin{align*}
N &=(q-1)^2(q^{m}-1)^2-(2q^{\frac{3m}2+1}-2q)(q^{\frac{m}4}-1)^2\\
&=(q^{\frac{m}4}-1)^2[(q-1)^2(1+q^{\frac{m}4}+q^{\frac{m}2}+q^{\frac{3m}4})^2-2q^{\frac{3m}2+1}+2q   ]\\
&> (q^{\frac{m}4}-1)^2[ (q-1)^2-2q] q^{\frac{3m}2}\\
&>0,
\end{align*}
a contradiction. Therefore, $d(\overline{\C}(-\bone))=4$ and $\overline{\C}(-\bone)$ is a distance-optimal code. 
This completes the proof. 	
\end{proof}

\begin{example}
We have the following examples of the code of Corollary \ref{OPM4:1}.
\begin{enumerate}
\item If $(q, m, e)=(2,6,1)$, the code $\overline{\C}(-\bone)$ has parameters $[64,54,4]$ and is distance-optimal.
\item If $(q, m, e)=(2,8,1)$, the code $\overline{\C}(-\bone)$ has parameters $[256,243,4]$ and is distance-optimal.
\item If $(q, m, e)=(3, 2,1)$, the code $\overline{\C}(-\bone)$ has parameters $[9,5,4]$ and is distance-optimal.
\item If $(q, m, e)=(3, 4,1)$, the code $\overline{\C}(-\bone)$ has parameters $[81,74,4]$ and is distance-optimal.  	
\end{enumerate}	
\end{example}

It is known that the extended code $\overline{\C}(-\bone)$ of a narrow-sense BCH code $\C$ of length $q^m-1$ over $\gf(q)$
is an affine-invariant code (see \cite{BC99}). Consequently, $d(\overline{\C}(-\bone))=d(\C)+1$. In general, we have the following results.  

\begin{theorem}\label{NTHM}
	Let $n>2$ be an integer with $\gcd(q,n)=1$ and $m=\ord_{n}(q)$. Let $\alpha\in \gf(q^m)$ be a primitive $n$-th root of unity. Let $\delta$ be an integer with $2\leq \delta<n$. Let $\C$ be the cyclic code of length $n$ over $\gf(q)$ with generator polynomial 
	$$\lcm(\m_{\alpha}(x), \m_{\alpha^2}(x), \cdots,  \m_{\alpha^{\delta-1}}(x)),$$
	where $\m_{\alpha^i}(x)$ denotes the minimal polynomial of $\alpha^i$ over $\gf(q)$. Then the following hold.
	\begin{enumerate}
	\item $d(\overline{\C}(-\bone))\geq \delta+1$ and $d(\overline{\C}(-\bone))=\delta+1$ provided that $d(\C)=\delta$.
	\item If $q^i\equiv -1\pmod{n}$ for some integer $i$ and $d(\C)\leq 2\delta-1$, then $$d(\overline{\C}(-\bone))=d(\C)+1.$$	
	\end{enumerate}
\end{theorem}

\begin{proof}
1. By the BCH bound, we have $d(\C)\geq \delta$. Consequently, $d(\overline{\C}(-\bone))=\delta$ if and only if $d(\C)=\delta$ and there is a codeword $c(x)=\sum_{i=1}^{\delta}c_i x^{j_i}\in \C$ with $\sum_{i=1}^{\delta}c_i=0$, where $c_i\in \gf(q)^*$, $j_a\neq j_b$ for $a\neq b$ and $0\leq j_i\leq n-1$. It is easily verified that $c(x)$ is a codeword of the cyclic code $\C'$ of length $n$ over $\gf(q)$ with generator polynomial 
	$$g'(x):=\lcm((x-1),\m_{\alpha}(x), \m_{\alpha^2}(x), \cdots,  \m_{\alpha^{\delta-1}}(x)).$$
	Since $\wt(c(x))=\delta$, we have $d(\C')\leq \delta$. By the BCH bound, $d(\C')\geq \delta+1$, a contradiction. Therefore, $d(\overline{\C}(-\bone))\geq \delta+1$. In particular, if $d(\C)=\delta$, then $$d(\overline{\C}(-\bone))\leq d(\C)+1=\delta+1.$$ Consequently, $d(\overline{\C}(-\bone))=\delta+1$. 
	
2. Suppose on the contrary that $d(\overline{\C}(-\bone))=d(\C)$. By the proof of Result 1, we obtain $d(\C')\leq d(\C)$. Since $q^i\equiv -1\pmod{n}$ for some integer $i$, we deduce that $g'(\alpha^i)=0$ for any $-(\delta-1)\leq i\leq \delta-1$. By the BCH bound, $d(\C')\geq 2\delta\geq d(\C)+1$, a contradiction. This completes the proof. 
	 \end{proof}
	 
\begin{corollary}\label{NCOR1}
Let $q=2^s$ with $s\geq 2$. Let $n=q^m+1$, where $m$ is a positive integer.	Let $\alpha\in \gf(q^{2m})$ be a primitive $n$-th root of unity and $\m_{\alpha}(x)$ denote the minimal polynomial of $\alpha$ over $\gf(q)$. Let $\C$ be the cyclic code of length $n$ over $\gf(q)$ with generator polynomial $\m_{\alpha}(x)$. Then $\overline{\C}(-\bone)$ has parameters $[n+1, n-2m, 4]$, and is distance-optimal.
\end{corollary}

\begin{proof}
It is clear that $d(\C)\geq 2$. If $d(\C)=2$, then there is an integer $i$ with $1\leq i\leq n-1$ such that $\alpha^{(q-1) i}=1$. Since $\gcd(q-1, n)=1$, we get $n\mid i$, a contradiction. Therefore, $d(\C)\geq 3$. If $d(\C)=3$, by Result 2 of Theorem \ref{NTHM}, we have $d(\overline{\C}(-\bone))=4$. If $d(\C)=4$, we have $d(\overline{\C}(-\bone))\geq d(\C)=4$.

It is clear that $\C$ has dimension $n-2m$. Then $\dim(\overline{\C}(-\bone))=n-2m$. Suppose there is an $\left[n+1, n-2m, d\geq 5\right]$ code over $\gf(q)$. By the sphere packing bound, we get  
\begin{align*}
1+\binom{n+1}{1} (q-1)+\binom{n+1}{2} (q-1)^2\leq q^{2m+1}.	
\end{align*}
Note that
\begin{align*}
\binom{n+1}{2} (q-1)^2&> \frac{(q-1)^2 q^{2m}	}{2} \\
&> q^{2m+1},
\end{align*}
a contradiction. Therefore, $d(\overline{\C}(-\bone))=4$ and $\overline{\C}(-\bone)$ is a distance-optimal code. This completes the proof.
\end{proof}

\begin{corollary}\label{NCOR2}
Let $q$ be a prime power, $m$ be a positive integer and $q^m\equiv 1 \pmod{4}$. Let $n=\frac{q^m+1}2$, $\alpha\in \gf(q^{2m})$ be a primitive $n$-th root of unity and $\m_{\alpha}(x)$ denote the minimal polynomial of $\alpha$ over $\gf(q)$. Let $\C$ be the cyclic code of length $n$ over $\gf(q)$ with generator polynomial $\m_{\alpha}(x)$. Then the following hold.
\begin{enumerate}
\item If $q\geq 11$, the code $\overline{\C}(-\bone)$ has parameters $[n+1, n-2m, 4]$, and is distance-optimal.	
\item If $q\in \{ 5, 7, 9\}$, the code $\overline{\C}(-\bone)$ has parameters $[n+1, n-2m, 4\leq d\leq 5]$ and $d=4$ provided that $n\equiv 0\pmod{3}$.
\item If $q=3$, the code $\overline{\C}(-\bone)$ has parameters $[n+1, n-2m, 5\leq d\leq 6]$, and $d=5$ provided that $m\equiv 4 \pmod{8}$ or $m\equiv 6 \pmod{12}$.
\end{enumerate}
\end{corollary}

\begin{proof}
Since $q^m\equiv 1 \pmod{4}$, we deduce that $n$ is odd and $\gcd(n, q-1)=1$. Similar to the proof of Corollary \ref{NCOR1}, we can prove that $d(\C)\geq 3$ and $d(\overline{\C}(-\bone))\geq 4$. It is clear that $\dim(\overline{\C}(-\bone))=n-2m$. 

1. Let $q\geq 11$. If there is an $\left[n+1, n-2m, d\geq 5\right]$ code over $\gf(q)$.  By the sphere packing bound, we get 
\begin{align*}
1+\binom{n+1}{1} (q-1)+\binom{n+1}{2} (q-1)^2\leq q^{2m+1}.	
\end{align*}
Note that
\begin{align*}
\binom{n+1}{2} (q-1)^2&> \frac{(q-1)^2 q^{2m}}{8} \\
&> q^{2m+1},
\end{align*}
a contradiction. Therefore, $d(\overline{\C}(-\bone))=4$ and $\overline{\C}(-\bone)$ is a distance-optimal code.

2. Let $q>3$. If there is an $[n+1, n-2m, d\geq 6]$ code over $\gf(q)$. By \cite[Theorem 1.5.1]{HP03}, there exists an $[n, n-2m, d\geq 5]$ code over $\gf(q)$. By the sphere packing bound, we get 
\begin{align*}
1+\binom{n}{1} (q-1)+\binom{n}{2} (q-1)^2\leq q^{2m}.	
\end{align*}
Note that
\begin{align*}
1+\binom{n}{1} (q-1)+\binom{n}{2} (q-1)^2&>\frac{(q-1)^2 q^{2m}}{8} \\
&> q^{2m},
\end{align*}
a contradiction. Therefore, $d(\overline{\C}(-\bone))\leq 5$. In particular, if $n\equiv 0\pmod{3}$, we have $1+\alpha^{\frac{n}3}+\alpha^{\frac{2n}3}=0$. Therefore, $d(\C)=3$. By  Theorem \ref{NTHM}, we get 
$$d(\overline{\C}(-\bone))=d(\C)+1=4.$$ 

3. Let $q=3$. It is easily verified that $\alpha^{1+2i}$ is a zero of $\m_{\alpha}(x)$ for each $-2\leq i\leq 1$. Since $\gcd(2, n)=1$, we get $d(\C)\geq 5$. Consequently, $d(\overline{\C}(-\bone))\geq 5$. By the sphere packing bound, a linear code over $\gf(q)$ with parameters $[n+1,n-2m, d\geq 7]$ does not exist. Therefore, $5\leq d(\overline{\C}(-\bone))\leq  6$. The rest of the proof is divided into the following cases.
\begin{enumerate}
\item[Case 1.] Let $m\equiv 4\pmod{8}$. It is easily verified that $\ord_{41}(3)=8$. Therefore, $3^{m}+1\equiv 3^4+1\equiv 0\pmod{41}$. Consequently, $41$ divides $n$. 	Let $\beta=\alpha^{n/41}$, then $\beta\in \gf(3^8)$ is a primitive $41$-th root of unity. Let $\C'$ be the cyclic code of length $41$ over $\gf(3)$ with generator polynomial $\m_{\beta}(x)$. It is easily verified that $d(\overline{\C'}(-\bone))=5$. Then there are five pairwise distinct integers $j_1,j_2,j_3,j_4, j_5$ with $0\leq j_i\leq 40$ and $c_i\in \gf(3)^*$ such that $\sum_{i=1}^5 c_i x^{j_i}\in \C'$ and $\sum_{i=1}^{5}c_i=0$. It then follows that 
$$\sum_{i=1}^5 c_i x^{(\frac{n}{41})j_i}\in \C. $$
Consequently, $d(\overline{\C}(-\bone))=5$.
\item[Case 2.] Let $m\equiv 6\pmod{12}$. It is easily checked that $\ord_{365}(3)=12$. Therefore, $3^{m}+1\equiv 3^6+1\equiv 0\pmod{365}$. Consequently, $365$ divides $n$. Let $\beta=\alpha^{n/365}$, then $\beta\in \gf(3^{12})$ is a primitive $365$-th root of unity. Let $\C'$ be the cyclic code of length $365$ over $\gf(3)$ with generator polynomial $\m_{\beta}(x)$. It is easily verified that $d(\overline{\C'}(-\bone))=5$. Then there are five pairwise distinct integers $j_1,j_2,j_3,j_4, j_5$ with $0\leq j_i\leq 364$ and $c_i\in \gf(3)^*$ such that $\sum_{i=1}^5 c_i x^{j_i}\in \C'$ and $\sum_{i=1}^{5}c_i=0$. It then follows that 
$$\sum_{i=1}^5 c_i x^{(\frac{n}{365})j_i}\in \C. $$
Consequently, $d(\overline{\C}(-\bone))=5$. 
\end{enumerate}
This completes the proof.	
\end{proof}

\begin{example}
We have the following examples of the code of Corollary \ref{NCOR2}.
\begin{enumerate}
\item If $(q, m)=(3,4)$, the code $\overline{\C}(-\bone)$ has parameters $[42,33,5]$, which are the best parameters known \cite{Grassl}.
\item If $(q, m)=(5,2)$, the code $\overline{\C}(-\bone)$ has parameters $[14,9,4]$ and is distance-optimal \cite{Grassl}.
\end{enumerate}
\end{example}
 
The parameters of the standardly extended code of the narrow-sense BCH code of length $n$ over $\gf(3)$ with distanced distance $4$ are documented in the next theorems.

\begin{theorem}\label{CMIM4}
	Let $n\geq 4$ be an integer with $\gcd(3,n)=1$ and $m=\ord_{n}(3)$. Let $\alpha\in \gf(3^m)$ be a primitive $n$-th root of unity. Let $\C$ be the cyclic code of length $n$ over $\gf(3)$ with generator polynomial $\lcm(\m_{\alpha}(x), \m_{\alpha^2}(x))$, where $\m_{\alpha^i}(x)$ denotes the minimal polynomial of $\alpha^i$ over $\gf(3)$. Then the following hold.
	\begin{enumerate}
	\item $d(\C)\geq 4$ and $d(\C)=4$ for $n\equiv 0\pmod{4}$.
	\item $d(\overline{\C}(-\bone))\geq 5$ and $d(\overline{\C}(-\bone))=5$ provided that $d(\C)=4$.	
	\end{enumerate}
\end{theorem}

\begin{proof}
It is clear that $\C$ is the narrow-sense BCH code of length $n$ over $\gf(3)$ with designed distance $4$. If $n\equiv 0\pmod{4}$, by \cite[Lemma 19]{LDL}, we deduce $d(\C)=4$. The second desired result follows directly from Theorem \ref{NTHM}. This completes the proof.
\end{proof}

\begin{corollary}\label{CEXT20}
Let $n=\frac{3^m-1}2$, where $m\geq 3$. Let $\C$ be the cyclic code of length $n$ over $\gf(3)$ in Theorem \ref{CMIM4}. Then the following hold.
\begin{enumerate}
\item The code $\C$ has parameters $[n, n-2m, 4]$.
\item The code $\overline{\C}(-\bone)$ has parameters $[n+1, n-2m, 5]$. It is both a dimension-almost-optimal code and a distance-almost-optimal code with respect to the sphere packing bound.	
\end{enumerate}
\end{corollary}

\begin{proof}
 When $m\geq 4$ is even, the cyclic code $\C$ over $\gf(3)$ was studied in \cite{SDXG}. According to \cite[Theorem 15]{SDXG}, $\C$ has parameters $[n, n-2m, 4]$. When $m\geq 3$ is odd, it is clear that $\C^\perp$ is the cyclic code of length $n$ over $\gf(3)$ with check polynomial $\m_{\alpha^{-1}}(x)\m_{\alpha^{-2}}(x)$. Let $\beta$ be a primitive element of $\gf(3^m)$ and $\beta^2=\alpha$. Then $\C^\perp$ has the following trace representation:
$$ \C^\perp=\{\bc(a, b):=(\tr_{3^m/3}( a \beta^{2i}+b \beta^{4i}) )_{i=0}^{n-1}: \ a , \ b \in \gf(3^m) \}.$$
Consider the cyclic code $\C'$ of length $2n=3^m-1$ over $\gf(3)$ with check polynomial $\m_{\beta^{-2}}(x)\m_{\beta^{-4}}(x)$. Then $\C'$ has the following trace representation:
$$ \C'= \{\bc'(a, b):=(\tr_{3^m/3}( a \beta^{2i}+b \beta^{4i}) )_{i=0}^{2n-1}: \ a , \ b \in \gf(3^m) \}.$$
For each $(a, b)\in \gf(3^m)^2$, we have $\bc'(a, b)=(\bc(a, b) \Arrowvert \bc(a, b))$, where $\Arrowvert$ denotes the concatenation of vectors. It follows that $\C$ has weight enumerator $W(z)$ if and only if the cyclic code $\C'$ has weight enumerator $W(z^2)$. The weight distribution of the cyclic code $\C'$ was determined in \cite{LF08}. When $m\geq 3$ is odd, according to \cite[Theorem 2]{LF08}, we deduce that $\C^\bot$ has weight enumerator 
\begin{align*}
	1+A_{w_1}^\perp z^{w_1}+A_{w_2}^\perp z^{w_2}+ A_{w_3}^\perp z^{w_3},
\end{align*}
where $w_1=3^{m-1}-3^{\frac{m-1}2}$, $A^\perp_{w_1}= (3^{\frac{m-1}2}+1)3^{\frac{m-1}2}n$, $w_2=3^{m-1}$, $A_{w_2}^\perp =2n(3^m-3^{m-1}+1)$, $w_3=3^{m-1}+3^{\frac{m-1}2}$ and $A_{w_3}^\perp =(3^{\frac{m-1}2}-1)3^{\frac{m-1}2} n$.

Let $A_i$ denote the total number of codewords weight $i$ in $\C$. Note that $d(\C)\geq 4$, then $A_i=0$ for $i\in \{1,2,3\}$. By the fifth Pless power moment, we get 
\begin{align*}
	w_1^4 A_{w_1}^\perp +w_2^4A_{w_2}^\perp +w_3^4 A_{w_3}^\perp &=3^{2m-4}(16n^4+48n^3-4n^2-6n+24A_4)\\
	&=3^{2m-4}(16n^4+48n^3-10n).
\end{align*}
It follows that 
$$A_4=\frac{n(n-1)}6>0.$$
Therefore, $\C$ has parameters $[n, n-2m, 4]$. By Theorem \ref{CMIM4}, $\overline{\C}(-\bone)$ has parameters $[n+1, n-2m, 5]$.

We now prove that $\overline{\C}(-\bone)$ is  distance-almost-optimal and  dimension-almost-optimal. Suppose there is an $[n+1, n-2m, d\geq 7]$ code over $\gf(3)$, by the sphere packing bound, we get 
$$1+\binom{n+1}{1} 2+\binom{n+1}{2} 2^2+\binom{n+1}{3} 2^3 \leq 3^{2m+1}=3(2n+1)^2.$$
It follows that $\frac{4n^3-30 n^2-28 n}{3}\leq 0$. Since $n=\frac{3^m-1}2\geq 13$, it is easily seen that $\frac{4n^3-30 n^2-28 n}{3}>0$, a contradiction.   Suppose there is an $[n+1,k\geq n+2-2m, 5]$ code over $\gf(3)$, by the sphere packing bound, we get 
$$1+\binom{n+1}{1} 2+\binom{n+1}{2} 2^2 \leq 3^{2m-1}=\frac{(2n+1)^2}{3}.$$
It follows that $\frac{2 n^2+8 n+ 8}{3}\leq 0$, a contradiction. This completes the proof. 
\end{proof}

\begin{example}
We have the following examples of the code of Corollary \ref{CEXT20}.
\begin{enumerate}
\item If $m=4$, the cyclic code $\C$ has parameters $[40,32, 4]$ and is a cyclic code with the best parameters \cite{DingBk1}. The code $\overline{\C}(-\bone)$ has parameters $[41, 32, 5]$, which are the best parameters known \cite{Grassl}.
\item If $m=5$, the cyclic code $\C$ has parameters $[121,111, 4]$. The best linear code over $\gf(3)$ of length $121$ and dimension $111$ has minimum distance $5$. The code $\overline{\C}(-\bone)$ has parameters $[122, 111, 5]$, which are the best parameters known \cite{Grassl}.
\end{enumerate}	
\end{example}

\begin{corollary}\label{CC-2}
Let $n=\frac{3^m-1}{e}$, where $m\geq 4$ is an even integer, $e$ is a divisor of $\frac{3^m-1}{4}$ and $4\leq e< \frac{3^{m/2}+1}2$. Let $\C$ be the cyclic code of length $n$ over $\gf(3)$ in Theorem \ref{CMIM4}. Then the following hold.
\begin{enumerate}
\item The code $\C$ has parameters $[n, n-2m, 4]$ and $\overline{\C}(-\bone)$ has parameters $[n+1, n-2m, 5]$.
\item When $m\geq 1+3\log_3e$, $\overline{\C}(-\bone)$ is a distance-almost-optimal code with respect to the sphere packing bound. 
\end{enumerate}
\end{corollary}

\begin{proof} Since $1\leq e<\frac{3^{m/2}+1}2$, we get $n>2(3^{\frac{m}2}-1)$. It follows that 
$$|C_1^{(3,n)}|=|C_2^{(3,n)}|=m.$$ 
Suppose $2\in C_1^{(3,n)}$, then there is an integer $j$ with $0\leq j\leq m-1$ such that $2\equiv 3^j\pmod{n}$. Since $n>2(3^{\frac{m}2}-1)$, we have $j>\frac{m}2$. Notice that $2\cdot 3^{m-j}\equiv 1\pmod{n}$, we have $n\leq 2\cdot 3^{m-j}-1<2\cdot 3^{\frac{m}2}-1$, a contradiction. Therefore, $\dim(\C)=n-2m$. Since $e$ divides $\frac{3^m-1}4$, we get that $n\equiv 0 \pmod{4}$. Consequently, $\C$ has parameters $[n, n-2m, 4]$. It is clear that $\overline{\C}(-\bone)$ has length $n+1$ and dimension $n-2m$. By Theorem \ref{CMIM4}, $d(\overline{\C}(-\bone))=5$.

We now prove that $\overline{\C}(-\bone)$ is a distance-almost-optimal code for $m\geq 1+3\log_3e$. Suppose there is an $[n+1, n-2m, d\geq 7]$ code over $\gf(3)$, by the sphere packing bound, we get 
$$1+\binom{n+1}{1} 2+\binom{n+1}{2} 2^2+\binom{n+1}{3} 2^3 \leq 3^{2m+1}=3(en+1)^2.$$
It follows that $\frac{4n^3-(9e^2-6)n^2-(18e-8)n}{3}\leq 0$. Since $m\geq 1+3\log_3e$, we get  $4n>9e^2$. Consequently, 
\begin{align*}
&\frac{4n^3-(9e^2-6)n^2-(18e-8)n}{3}\\
&> \frac{6n^2-(18e-8)n}{3} >0,
\end{align*}
a contradiction. This completes the proof. 
\end{proof}

\begin{example}
 Let $(m, e)=(4, 4)$, and let $\C$ be the cyclic code in Corollary \ref{CC-2}. Then the code $\overline{\C}(-\bone)$ has parameters $[21, 12, 5]$ and is distance-almost-optimal \cite{Grassl}.
\end{example}

\begin{corollary}\label{CEX25}
Let $n=2(3^s-1)$, where $s\geq 2$ is an integer. Let $\C$ be the cyclic code of length $n$ over $\gf(3)$ in Theorem \ref{CMIM4}. Then the following hold.
\begin{enumerate}
\item The code $\C$ has parameters $[n, n-3s, 4]$.
\item The extended code $\overline{\C}(-\bone)$ has parameters $[n+1, n-3s, 5]$ and is distance-almost-optimal with respect to the sphere packing bound.	 	
\end{enumerate}
\end{corollary}

\begin{proof}
It is easily verified that $\deg(\m_{\alpha}(x))=2s$ and $\deg(\m_{\alpha^2}(x))=s$. Consequently, $\dim(\C)=n-3s$. Note that $4$ divides $n$, we have $d(\C)=4$. Therefore, $\C$ has parameters $[n, n-3s, 4]$. It follows from Theorem \ref{CMIM4} that $\overline{\C}(-\bone)$ has parameters $[n+1, n-3s,5]$. 	

We now prove that $\overline{\C}(-\bone)$ is a distance-almost-optimal code. Suppose there is an $[n+1, n-3s, d\geq 7]$ code over $\gf(3)$, by the sphere packing bound, we get
$$1+\binom{n+1}{1} 2+\binom{n+1}{2} 2^2+\binom{n+1}{3} 2^3 \leq 3^{3s+1}=3(\frac{n+2}2)^3.$$
It follows that $\frac{23 n^3-6n^2-44n}{24}\leq 0$. Note that $n\geq 16$, we have $23 n^3-6n^2-44n>0$, a contradiction. This completes the proof. 
\end{proof}

\begin{example}
	We have the following examples of the code of Corollary \ref{CEX25}.
	\begin{enumerate}
	\item If $s=2$, the code $\overline{\C}(-\bone)$ has parameters $[17, 10, 5]$ and is distance-optimal \cite{Grassl}.	
	\item If $s=3$, the code $\overline{\C}(-\bone)$ has parameters $[53, 43, 5]$, which are the best parameters known \cite{Grassl}.
	\end{enumerate}
\end{example}

\begin{corollary}\label{CEX26}
Let $n=2(3^s+1)$, where $s\geq 2$ is an integer. Let $\C$ be the cyclic code of length $n$ over $\gf(3)$ in Theorem \ref{CMIM4}. Then  $\C$ has parameters $[n, n-4s, 4]$ and $\overline{\C}(-\bone)$ has parameters $[n+1, n-4s, 5]$.	
\end{corollary}
\begin{proof}
It is easily verified that $|C_1^{(3,n)}|=|C_2^{(3,n)}|=2s$. Therefore, $\dim(\C)=n-4s$. Thereby, $$\dim(\overline{\C}(-\bone))=n-4s.$$ Note that $n\equiv 0\pmod{4}$, by Theorem \ref{CMIM4}, $d(\C)=4$ and $d(\overline{\C}(-\bone))=5$. This completes the proof. 	
\end{proof}

\begin{example}
	We have the following examples of the code of Corollary \ref{CEX26}.
	\begin{enumerate}
	\item If $s=1$, the code $\overline{\C}(-\bone)$ has parameters $[9, 4, 5]$ and is distance-optimal \cite{Grassl}.	
	\item If $s=2$, the code $\overline{\C}(-\bone)$ has parameters $[21, 12, 5]$ and is distance-almost-optimal \cite{Grassl}.
	\end{enumerate}
\end{example}

\section{The extended codes of nonbinary Hamming codes} \label{Sec:6}

Let $m \geq 2$ be an integer and let $q$ be a prime power. Put $n=\frac{q^m-1}{q-1}$ and 
\begin{eqnarray}\label{eqn-alphai}
\alpha_i=a_i  \alpha^i 
\end{eqnarray}
for all $0 \leq i \leq n-1$, where $\alpha$ is a primitive element of $\gf(q^m)$ and the vector  
$$\ba=(a_0, a_1, \ldots, a_{n-1}) \in (\gf(q)^*)^n.$$ 
It is easily seen that any two distinct elements in the set $\{\alpha_i: 0 \leq i \leq n-1\}$ are linearly independent over $\gf(q)$, 
and the set $\{\alpha_i: 0 \leq i \leq n-1\}$ is a point set of the projective space $\PG(m-1, \gf(q))$ if we 
identify $(\gf(q^m),+)$ with $(\gf(q)^m, +)$. 

For each $\ba=(a_0, a_1, \ldots, a_{n-1}) \in (\gf(q)^*)^n$, a Hamming code $\cHam(q,m, \ba)$ is defined to be 
\begin{eqnarray}\label{eqn-HammingCodea}
\cHam(q,m, \ba)=\left\{\bc=(c_0,c_1, \ldots, c_{n-1}) \in \gf(q)^n: \sum_{i=0}^{n-1} c_i\alpha_i=0\right\}. 
\end{eqnarray}
Since there are $(q-1)^n$ choices of $\ba$, we have defined $(q-1)^n$ Hamming codes $\cHam(q,m, \ba)$ 
here. Any two of them are scalar-equivalent. In addition, any Hamming code defined in other ways must be 
permutation-equivalent to a Hamming code $\cHam(q,m, \ba)$ defined above.  It is well known that every 
Hamming code $\cHam(q,m, \ba)$ has parameters $[n, n-m, 3]$ and its weight distribution is known and independent 
of $\ba$. The dual code of the Hamming code is called the Simplex code and denoted by $\cSim(q, m, \ba)$, which 
has parameters $[n,m, q^{m-1}]$ and is a one-weight code. Some Hamming codes are cyclic and some are constacyclic. The following theorem documents a class of constacyclic Hamming codes.

\begin{theorem}\label{thm-12}
Let $\lambda=\alpha^n$ and let $\m_{\alpha}(x)$ be the minimal polynomial of $\alpha$ over $\gf(q)$. Then $\cHam(q,m, \ba)$ is scalar-equivalent to $\cHam(q,m, \bone)$, and $\cHam(q,m, \bone)$ is the $\lambda$-constacyclic code of length $n$ with generator polynomial $\m_{\alpha}(x)$.
\end{theorem}

\begin{proof}
It is easily verified that 
\begin{align}\label{EEQQ-21}
\cHam(q, m, \ba) =~\{(a_0^{-1}  c_0, a_1^{-1}  c_1,\ldots, a_{n-1}^{-1} c_{n-1}): \ (c_0, c_1,\ldots, c_{n-1}) \in  \cHam(q, m, \bone) \}.	
\end{align}
The first desired result follows.

Let $\C$ be the $\lambda$-constacyclic code of length $n$ with generator polynomial $\m_{\alpha}(x)$. It is easily verified that $\dim(\C)=n-m$. Suppose $\bc=(c_0,c_1,\ldots, c_{n-1})\in \cHam(q,m, \bone)$. By definition, 
$$c_0+c_1\alpha+\cdots+c_{n-1}\alpha^{n-1}=0.$$ 
It follows that $\bc\in \C$. Consequently, $\cHam(q,m, \bone)\subseteq \C$. Note that $\dim(\cHam(q,m, \bone))=\dim(\C)$. The second desired result follows. 
\end{proof}

The standardly extended code $\overline{\cHam(2,m, \bone)}(-\bone)$ of the binary Hamming code $\cHam(2,m, \bone)$ was studied in the literature and is known to have the parameters 
$[2^m, 2^m-1-m, 4]$. The objective of this section is to study the standardly extended code $\overline{\cHam(q,m, \ba)}(-\bone)$ of the nonbinary Hamming code $\cHam(q,m, \ba)$. The length and dimension of $\overline{\cHam(q,m, \ba)}(-\bone)$ are obvious. Hence, we investigate the minimum distance of $\overline{\cHam(q,m, \ba)}(-\bone)$. 

\begin{theorem}\label{thm-Hamcodemain}
For each $\ba=(a_0, a_1, \ldots, a_{n-1}) \in (\gf(q)^*)^n$,  define 
\begin{eqnarray}
S(q,m, \ba)=\left\{  \frac{\frac{a_{i_3}}{a_{i_1}}  \alpha^{i_3-i_1}-1}{\frac{a_{i_2}}{a_{i_1}} \alpha^{i_2-i_1}-1 }: 0 \leq i_1 < i_2 <i_3 \leq n-1  \right\}. 
\end{eqnarray} 
Then 
$ 
d(\overline{\cHam(q,m, \ba)}(-\bone))=3  
$ 
if and only if $|S(q,m, \ba) \cap \gf(q)^*| \geq 1$  
and  
$ 
d(\overline{\cHam(q,m, \ba)}(-\bone))
$ 
$
=4 
$ 
if and only if $|S(q,m, \ba) \cap \gf(q)^*|=0$. 
\end{theorem}

\begin{proof}
By definition, $d(\overline{\cHam(q,m, \ba)}(-\bone))=3$ if and only if there are $a, b, c$ in $\gf(q)^*$ and integers 
$i_1, i_2, i_3$ such that $0 \leq i_1<i_2<i_3 \leq n-1$ and 
\begin{eqnarray*}
\left\{ 
\begin{array}{l}
a a_{i_1}  \alpha^{i_1}+b  a_{i_2} \alpha^{i_2}+c  a_{i_3} \alpha^{i_3} =0, \\
a+b+c=0, 
\end{array}
\right. 
\end{eqnarray*} 
which is the same as 
\begin{eqnarray*}
\left\{ 
\begin{array}{l}
\frac{a}{c}=\frac{\frac{a_{i_3}}{a_{i_1}} \alpha^{i_3-i_1}-1}{\frac{a_{i_2}}{a_{i_1}} \alpha^{i_2-i_1}-1 }-1, \\
1+\frac{b}{a}+\frac{c}{a}=0.
\end{array} 
\right. 
\end{eqnarray*}
The desired conclusions then follow. 
\end{proof} 

\begin{example}\label{exam-HCcountere1} 
Let $q=4$ and let $m=2$. Then $n=q+1=5$. Let $\alpha$ be the primitive element of $\gf(q^m)$ with $\alpha^4+\alpha+1=0$. 
Define $\beta=\alpha^n$. Then $\beta$ is a primitive element of $\gf(q)$. 
When $\ba=(1, \beta,\beta^2, 1, \beta)$ or 
$\ba=(\beta^2, 1, \beta, \beta^2, 1)$, it can be verified by Magma that $|S(q,m, \ba) \cap \gf(q)^*| =0$. It then follows 
from Theorem \ref{thm-Hamcodemain} that $d(\overline{\cHam(q,m, \ba)}(-\bone))=4$. 
When $\ba=(\beta,\beta^2, \beta, 1, \beta^2)$, it can be verified by Magma that $|S(q,m, \ba) \cap \gf(q)^*| =2$. It then follows from Theorem \ref{thm-Hamcodemain} that $d(\overline{\cHam(q,m, \ba)}(-\bone))=3$.
\end{example}

\begin{theorem}\label{MDS::1}
Let $q=2^s$ with $s\geq 2$. Let $\alpha$ be a primitive element of $\gf(q^2)$ and $\lambda=\alpha^{q+1}$. Let $$\ba=(1,\lambda^{2^{s-1}-1},\ldots,\lambda^{i (2^{s-1}-1)}, \ldots, \lambda^{q (2^{s-1}-1)}).$$ 
Then $d(\overline{\cHam(q, 2, \ba)}(-\bone))=4$.
\end{theorem}

\begin{proof}
It is clear that $\lambda$ is a primitive element of $\gf(q)$. Thus, $\lambda^{(2^{s-1}-1)(q+1)}=\lambda^{q-2}$. Let $$(c_0,c_1, \ldots, c_{q})\in \cHam(q, 2, \ba).$$
Then 
\begin{align*}
\sum_{i=0}^{q} c_i \lambda^{i (2^{s-1}-1)} \alpha^i=\sum_{i=0}^{q} c_i (\lambda^{2^{s-1}-1} \alpha)^i.
\end{align*} 
Define $\beta=\lambda^{2^{s-1}-1} \alpha$. 
Note that 
\begin{align*}
\beta^{q+1}&= (\lambda^{2^{s-1}-1} \alpha)^{q+1}\\
&=\lambda^{(2^{s-1}-1)(q+1)}\alpha^{q+1}	\\
&=\lambda^{q-2} \lambda \\
&=1.
\end{align*}
Similar to the proof of Theorem \ref{thm-12}, we can prove that $\cHam(q, 2, \ba)$ is the cyclic code of length $q+1$ over $\gf(q)$ with generator polynomial $\m_{\beta}(x)$. Since $d(\cHam(q, 2, \ba))=3$, we deduce that $\beta$ is a primitive $(q+1)$-th root of unity. Let $\theta=\beta^{-1}$, then $\cHam(q,2, \ba)^{\perp}$ has the following trace representation:
$$\left\{\bc(b)=(\tr_{q^2/q}(b),\tr_{q^2/q}(b\theta), \ldots, \tr_{q^2/q}(b\theta^q)): \ b\in \gf(q^2) \right\}. $$
By Theorem \ref{thm:1}, we get 
$$\overline{\cHam(q, 2, \ba)}(-\bone)^\perp=\left\{ \bc(b, a)=(\bc(b)+a \bone, a): \ b\in \gf(q^2),~ a\in \gf(q) \right\}.$$
The rest of the proof is divided into the following cases:
\begin{enumerate}
\item[Case 1.] If $a=b=0$, we have $\wt(\bc(b,a))=0$.
\item[Case 2.] If $a=0$ and $b\neq 0$, we have $\wt(\bc(b,a))=q$.
\item[Case 3.] If $a\neq 0$ and $b=0$, we have $\wt(\bc(b,a))=q+2$.
\item[Case 4.] If $a\neq 0$ and $b\neq 0$, we have $\wt(\bc(b, a))=q+2-N(a, b)$, where 
\begin{align*}
N(a,b)&=|\{i: 0\leq i\leq q, ~ \tr_{q^2/q}(b \theta^i )=a \} |.
\end{align*}
It is clear that 
\begin{align*}
\tr_{q^2/q}(b \theta^i)=a	&\iff b^q \theta^{qi}+b \theta^i=a\\
& \iff b^q \theta^{-i}+b \theta^i=a\\
& \iff \theta^{-i}(b^q+b \theta^{2i}+a \theta^i)=0.\\
& \iff b^q+b \theta^{2i}+a \theta^i=0. 
\end{align*}
It follows that $N(a, b)\leq 2$. Consequently, $\wt(\bc(b, a))=q+2-N(a, b) \geq q$. 
\end{enumerate}

Summarizing the conclusions of the four cases above, we deduce that $d(\overline{\cHam(q, 2, \ba)}(-\bone)^\perp)\geq q$. It is clear $\overline{\cHam(q, 2, \ba)}(-\bone)^\perp$ has length $q+2$ and dimension $3$. By the Singleton bound, we have $$d(\overline{\cHam(q, 2, \ba)}(-\bone)^\perp)\leq q.$$ Therefore, $\overline{\cHam(q, 2, \ba)}(-\bone)^\perp$ is a $[q+2, 3, q]$ MDS code over $\gf(q)$. Consequently, $\overline{\cHam(q, 2, \ba)}(-\bone)$ is a $[q+2, q-1, 4]$ MDS code over $\gf(q)$. This completes the proof.  
\end{proof}

Example \ref{exam-HCcountere1} and Theorem \ref{MDS::1} show that the standardly extended code of some Hamming codes 
has minimum distance $4$. However, we have the following conjecture. 

\begin{conj}\label{conj-HCextend} 
Let $m \geq 2$ and $q \geq 3$. For each $\ba=(a_0, a_1, \ldots, a_{n-1}) \in (\gf(q)^*)^n$, we have  
$$
d(\overline{\cHam(q,m, \ba)}(-\bone))=3, 
$$ 
provided that $(q, m)\neq (2^s, 2)$ with $s \geq 2$. 
\end{conj}   

It looks hard to prove or disprove the conjecture above. The reader is cordially invited to attack this problem. Below we settle this conjecture in some cases.

\begin{theorem}\label{thm15}
Let $q$ be odd and $m\geq 2$ be even. Then $d(\overline{\cHam(q,m, \ba)}(-\bone))=3$. 
\end{theorem}

\begin{proof}
It is clear that $d(\overline{\cHam(q,m, \ba)}(-\bone))\geq 3$. We now prove that $d(\overline{\cHam(q,m, \ba)}(-\bone))\leq 3$. To prove the desired conclusion, we only need to prove that there is a codeword 
$$\bc=(c_0,c_1,\ldots,c_{n-1})\in \cHam(q,m, \ba)$$ 
with $\wt(\bc)=3$ and $\sum_{i=0}^{n-1} c_i=0$.  

Since $m$ is even, we deduce that $(q^2-1)\mid (q^m-1)$. Consequently, $q+1$ divides $n$. Let $\overline{n}=\frac{n}{q+1}$ and  
$$\C_1=\left\{ \bc=(c_{0},c_{\overline{n}},\ldots,c_{\overline{n} q})\in \gf(q)^{q+1}: \sum_{i=0}^{q} c_{\overline{n} i } \alpha_{\overline{n} i }=0  \right\}. $$
Notice that $\alpha^{\overline{n}}$ is a primitive element of $\gf(q^2)$. Let $$\ba'=(a_0,a_{\overline{n}}, a_{\overline{n}2 }, \ldots, a_{\overline{n}q})\in (\gf(q)^*)^{q+1},$$
then $\C_1$ is the Hamming code $\cHam(q, 2, \ba')$, and $\C_1$ has parameters $[q+1,q-1,3]$.

Let $\bone'$ denote the all-one vector over $\gf(q)$ with length $q+1$. Now consider the extended code $\overline{\C_1}(-\bone')$. It is easy to see that $\dim(\overline{\C_1}(-\bone'))=\dim(\C_1)=q-1$ and $d( \overline{\C_1}(-\bone') )= 3$ or $4$. If $$d(\overline{\C_1}(-\bone')) = 4,$$ then $\overline{\C_1}(-\bone')^\bot $ is a $[q+2, 3, q]$ MDS code over $\gf(q)$. This contradicts the conclusion of Lemma \ref{lem:9}. Therefore, $d(\overline{\C_1}(-\bone'))=3$. It follows that there is a codeword $\bc'=(c' _{0},c'_{\overline{n}},\ldots,c'_{\overline{n} q})\in \C_1$ with $\wt(\bc')=3$ and $\sum_{i=0}^{q} c' _{\overline{n} i}=0$.

Let $\bc=(c_0,c_1,\ldots, c_{n-1})$, where $c_i=0$ if $i\not\equiv 0 \pmod{\overline{n}}$, and $ c_{\overline{n} i}= c' _{\overline{n} i}$ for each $i\in \{0,1,\cdots, q\}$. It is clear that 
\begin{align*}
	\sum_{i=0}^{n-1}c_i \alpha_i=\sum_{i=0}^{q} c_{\overline{n} i}\alpha_{\overline{n}i}=0.
\end{align*}
Therefore, $\bc \in \cHam(q, m, \ba)$. Note that  $\wt(\bc)=3$ and 
$$\sum_{i=0}^{n-1}c_i=\sum_{i=0}^{q} c' _{\overline{n}  i}=0.$$
It then follows that $d(\overline{\cHam(q,m, \ba)}(-\bone))\leq 3$. This completes the proof. 
\end{proof}

\begin{theorem}\label{thm16}
Let $q>2$ be a prime power and $m\equiv 0\pmod{3}$. Then  $d(\overline{\cHam(q, m, \ba)}(-\bone))=3$. 
\end{theorem}

\begin{proof}
Similar to Theorem \ref{thm15}, we only need to prove that there is a codeword 
$$\bc=(c_0,c_1,\ldots,c_{n-1})\in \cHam(q,m, \ba)$$ 
with $\wt(\bc)=3$ and $\sum_{i=0}^{n-1} c_i=0$.  
  
Since $m\equiv 0\pmod{3}$, it follows that $(q^3-1)\mid (q^m-1)$. Consequently, $q^2+q+1$ divides $n$. Let $\overline{n}=\frac{n}{q^2+q+1}$, and  
$$\C_1=\left\{ (c_{0},c_{\overline{n}},\ldots,c_{\overline{n} (q^2+q) })\in \gf(q)^{q^2+q+1}: \sum_{i=0}^{q^2+q} c_{\overline{n} i} \alpha_{\overline{n} i}=0  \right\}. $$
Notice that $\alpha^{\overline{n}}$ is a primitive element of $\gf(q^3)$. Let $$\ba'=(a_0,a_{\overline{n}}, a_{\overline{n}2 }, \ldots, a_{\overline{n}(q^2+q)})\in (\gf(q)^*)^{q^2+q+1},$$
then $\C_1$ is the Hamming code $\cHam(q, 3, \ba')$, and $\C_1$ has parameters $[q^2+q+1,q^2+q-2,3]$.
 
Let $\bone'$ denote the all-one vector over $\gf(q)$ with length $q^2+q+1$. Now consider the extended code $\overline{\C_1}(-\bone')$. It is easily seen that 
$$\dim(\overline{\C_1}(-\bone'))=\dim(\C_1)=q^2+q-2$$
 and $d( \overline{\C_1}(-\bone') )=3$ or $4$. If $d( \overline{\C_1}(-\bone'))=4$, then $\overline{\C_1}(-\bone')$ is a $[q^2+q+2, q^2+q-2, 4]$ AMDS code over $\gf(q)$. According to Lemma \ref{AMDS}, a linear code over $\gf(q)$ with parameters $[q^2+q+2, q^2+q-2, 4]$ does not exist. Therefore, $d(\overline{\C_1}(-\bone'))=3$. It follows that there is a codeword $\bc'=(c' _{0},c'_{\overline{n}},\ldots,c'_{\overline{n}(q^2+q)})\in \C_1$ with $\wt(\bc')=3$ and $$\sum_{i=0}^{q^2+q} c'_{\overline{n} i }=0.$$

Let $\bc=(c_0,c_1,\ldots, c_{n-1})$, where $c_i=0$ if $i\not\equiv 0 \pmod{\overline{n}}$, and $ c_{\overline{n}  i}= c' _{\overline{n}  i}$ for each $i\in \{0,1,\cdots, q^2+q\}$. It is clear that 
\begin{align*}
	\sum_{i=0}^{n-1}c_i \alpha_i=\sum_{i=0}^{q^2+q} c_{\overline{n} i}\alpha_{\overline{n} i}=0.
\end{align*}
Therefore, $\bc \in \cHam(q, m, \ba)$. Note that $\wt(\bc)=3$ and 
$$\sum_{i=0}^{n-1}c_i=\sum_{i=0}^{q^2+q} c' _{\overline{n} i}=0.$$
It then follows that $d(\overline{\cHam(q,m, \ba)}(-\bone))\leq 3$. This completes the proof. 
\end{proof}

The minimum distance of the standardly extended codes of constacyclic Hamming codes $\cHam(q, m, \bone)$  is documented in the following theorem. 

\begin{theorem}\label{HAMM}
Let $q>2$ be a prime power and let $m\geq 2$. Then $d(\overline{\cHam(q, m, \bone)}(-\bone))=3$.
\end{theorem}

\begin{proof}
Recall that $\alpha$ is a primitive element of $\gf(q^m)$, $n=\frac{q^m-1}{q-1}$ and $\lambda=\alpha^n$, then $\lambda$ is a primitive element of $\gf(q)$. Similar to Theorem \ref{thm15}, we only need to prove that there are distinct integers $j_0, j_1, j_2$ with $0\leq j_i\leq n-1$ and $c\in \gf(q)\backslash \{0,1\}$ such that $(1-c) x^{j_0}+cx^{j_1}-x^{j_2}\in \cHam(q, m, \bone)$. 

By Theorem \ref{thm-12}, $\cHam(q, m, \bone)$ is the $\lambda$-constacyclic code of length $n$ over $\gf(q)$ with generator polynomial $\m_{\alpha}(x)$. Since $d(\C)=3$, we deduce $\alpha-\alpha^i \neq 0$ for each $i\in \{0,2, \cdots, n-1 \}$. The rest of the proof is divided into the following cases. 

\begin{enumerate}
\item[Case 1.] If there are $i \in \{0,2, \cdots, n-1\}$ and $c\in \gf(q)^*$ such that $\alpha-\alpha^i= c \alpha$, then $$(1-c)\alpha-\alpha^i=0,$$ 
which contradicts the fact that $d(\cHam(q, m, \bone))=3$.
\item[Case 2.] If there are $i, j \in \{0,2,\cdots, n-1\}$ with $i\neq j$ and $c\in \gf(q)^*$ such that $$\alpha-\alpha^i=c( \alpha-\alpha^j),$$ then $(1-c)\alpha-\alpha^i+c \alpha^j=0$. Since $d(\cHam(q, m, \bone))=3$, we have $c\neq 1$. Consequently, 
$$(1-c)x-x^i+c x^j\in \cHam(q, m, \bone).$$ 
\item[Case 3.] If $\alpha-\alpha^i \neq c (\alpha-\alpha^j)$ for any $c\in \gf(q)$ and $i, j \in \{0,2,\cdots, n-1\}$ with $i\neq j$, then $\{\alpha, \alpha-1, \alpha-\alpha^2, \cdots, \alpha-\alpha^{n-1}\}$ is the set of projective points in $\PG(1,\gf(q))$. Since $d(\cHam(q, m, \bone))=3$ and $\lambda\notin \{0,1\}$, we deduce $\alpha-\lambda \neq c \alpha$ for any $c\in \gf(q)^*$, and $\alpha-\lambda \neq c(\alpha-1)$ for any $c\in \gf(q)^*$. Therefore, there are $i$ with $2\leq i\leq n-1$ and $c\in \gf(q)^*$ such that $\alpha-\alpha^{i}=c(\alpha-\lambda)$. It follows that 
\begin{equation}\label{EEQQ48}
c(x):=c\lambda+(1-c)x-x^i\in \cHam(q, m, \bone).
\end{equation}
Since $d(\cHam(q, m, \bone))=3$, we have $c\neq 1$. It follows from (\ref{EEQQ48}) that 
\begin{align*}
x^{n-1} c(x)&\equiv c\lambda x^{n-1}+(1-c)\lambda -\lambda x^{i-1}\pmod{x^n-\lambda}.
\end{align*}
It then follows that 
$$(1-c)-x^{i-1}+c x^{n-1}\in \cHam(q, m, \bone).$$ 
\end{enumerate}
This completes the proof.
\end{proof}

\begin{example}
We have the following examples of the code of Theorem \ref{HAMM}.
\begin{enumerate}
\item If $(q,m)=(3, 2)$, the code $\overline{\cHam(q, m, \bone)}(-\bone)$ has parameters $[5, 2, 3]$ and is  distance-optimal \cite{Grassl}.	
\item If $(q,m)=(3, 3)$, the code $\overline{\cHam(q, m, \bone)}(-\bone)$ has parameters $[14, 10, 3]$ and is 
distance-optimal \cite{Grassl}.	
\end{enumerate}	
\end{example}

The parameters of the standardly extended codes of $\cSim(q, m, \ba)$ are documented in the next theorem. 

\begin{theorem}\label{thm-nov251}
Let $\ba \notin \cHam(q,m, \bone) $. Then $\overline{\cSim(q, m, \ba)}(-\bone)$ has parameters 
$$\left[\frac{q^m+q-2}{q-1}, m, q^{m-1}\right]$$
and weight enumerator
$$1+(q^{m-1}-1)z^{q^{m-1}}+q^{m-1}(q-1)z^{q^{m-1}+1}.$$ 
The dual code $\overline{\cSim(q, m, \ba)}(-\bone)^{\bot}$ has parameters $$\left[\frac{q^m+q-2}{q-1},\frac{q^m+q-2}{q-1}- m, 2\right],$$ 
and is distance-optimal.
\end{theorem}

\begin{proof}
 It is easily verified that $-\1 \in \cHam(q,m, \ba)$ if and only if $\ba \in \cHam(q,m, \bone)$.  Therefore, if $\ba \notin \cHam(q,m, \bone)$, we have $-\1 \notin \cHam(q,m, \ba)$. It follows from Theorem \ref{lem:2} that
  $$d(\overline{\cSim(q, m, \ba)}(-\bone)^{\bot})\geq 2.$$ 
Since $\cHam(q, m, \ba)^{\bot}$ has weight enumerator $1+(q^m-1)z^{q^{m-1}}$, the possible nonzero weights of $$\overline{\cSim(q, m, \ba)}(-\bone)$$ are $q^{m-1}$ and $q^{m-1}+1$. Let $w_1=q^{m-1}$ and $w_2=q^{m-1}+1$. Then
\begin{equation}\label{EQ9}
A_{w_1}+A_{w_2}=q^m-1.	
\end{equation}
Since $d(\overline{\cSim(q, m, \ba)} (-\bone)^{\bot})\geq 2$, by the second Pless power moment, we have  
\begin{equation}\label{EQ10}
w_1A_{w_1}+w_2A_{w_2}=q^{m-1}(q^m+q-2).	
\end{equation}
It follows from (\ref{EQ9}) and (\ref{EQ10}) that 
\begin{align*}
A_{w_1}&=q^{m-1}-1,\\
A_{w_2}&=q^{m-1}(q-1).	
\end{align*}
By the third Pless power moment, we can get $A_2^{\bot}=q-1$. By the sphere packing bound, a linear code over $\gf(q)$ with parameters $[n+1, n+1-m, d\geq 3]$ does not exist. This completes the proof. 
\end{proof}

It is easily seen that
\begin{align*}
\gcd(1+x+\cdots+x^{n-1},x^n-\lambda)=1	
\end{align*}
for $\lambda\neq 1$. Then $\1 \notin \cHam(q,m, \bone)$. The results in the following theorem then follow from Theorem \ref{thm-nov251}. 

\begin{corollary}
Let $q>2$. Then the extended code $\overline{\cSim(q, m, \bone)}(-\bone)$ has parameters 
$$\left[\frac{q^m+q-2}{q-1}, m, q^{m-1}\right]$$
and weight enumerator 
 $$1+(q^{m-1}-1)z^{q^{m-1}}+q^{m-1}(q-1)z^{q^{m-1}+1}.$$
 The dual code $\overline{\cSim(q, m, \bone)}(-\bone)^{\bot}$ has parameters $$\left[\frac{q^m+q-2}{q-1},\frac{q^m+q-2}{q-1}- m, 2\right],$$
and is distance-optimal.
\end{corollary}

\section{Summary of contributions and concluding remarks}\label{Sec:7}

It was demonstrated in Example \ref{exam-HCcountere1}  that $\overline{\C_1}(-\bone)$ and $\overline{\C_2}(-\bone)$ could be very different even if 
$\C_1$ and $\C_2$ are scalar-equivalent. The determination of the minimum distance of an extended code $\overline{\C}(\bu)$ 
is a very challenging problem in general, as it requires a lot of information on all the minimum weight codewords in $\C$. 
These may partially explain why the minimum distance of the sdandardly extended Hamming code 
$\overline{\cHam(q, m, \ba)}(-\bone)$ for $q >2$ remained open for over 70 years. 
As demonstrated 
earlier, the minimum distance $d(\overline{\cHam(q, m, \ba)}(-\bone))$ of the sdandardly extended Hamming code 
$\overline{\cHam(q, m, \ba)}(-\bone)$ could be $3$ or $4$. 
The reader is cordially invited to work on Conjecture \ref{conj-HCextend}.   

In this paper, several fundamental results about the extended codes $\overline{\C}(\bu)$ of linear codes over finite fields were derived and presented in Theorem \ref{thm-general111}, Theorem \ref{thm-general112} and Corollary \ref{thm-general113}. The dual of the extended codes was characterized in Theorem \ref{thm:1}, and the dual distance of the extended codes was analyzed in Theorem \ref{dual-distance} and Theorem \ref{lem:2}. A lower bound of the minimum distance of the extended codes of narrow-sense nonprimitive BCH codes was developed in Theorem \ref{NTHM}. The parameters of the sdandardly extended code of some nonbinary Hammming codes and some Simplex codes were investigated in Section \ref{Sec:6}. Many infinite families of optimal codes were obtained in this paper with the extending technique $\overline{\C}(\bu)$.

A summary of the major specific contributions of this paper goes as follows.
\begin{enumerate}
\item Some extended codes of a family of $[q, 3, q-2]$ MDS codes $\C_\alpha$ over $\gf(q)$ were studied in Section \ref{sec-nov24}, and the following main results were obtained:  
\begin{itemize}
\item Although the standardly extended code $\overline{\C_\alpha}(-\bone)$ is trivial,  there is a vector $\bu^{(1)}$ such that the extended code $\overline{\C_\alpha}(\bu^{(1)})$ is a $[q+1, 3, q-1]$ MDS code over $\gf(q)$ (see Theorem \ref{INF-TT}).	
\item The doubly extended code $\overline{\overline{\C_\alpha}(\bu^{(1)})}(-\bone)$ is a $[q+2,3,q-1]$ AMDS code over $\gf(q)$ and $d(\overline{\overline{\C_\alpha}(\bu^{(1)})}(-\bone)^\perp)=2$. This doubly extended code $\overline{\overline{\C_\alpha}(\bu^{(1)})}(-\bone)$ is not very special. However, 
there is a vector $\bu^{(2)}$ such that the following hold: 
\begin{itemize}
\item The doubly extended code $\overline{\overline{\C_\alpha}(\bu^{(1)})}(\bu^{(2)})$ is a $[q+2, 3, q-1]$ NMDS code over $\gf(q)$ for odd $q\geq 5$ and a $[q+2, 3, q]$ MDS code over $\gf(q)$ for even $q\geq 4$ (see Theorem \ref{ConicCode1}). Hence, an infinite family of NMDS codes and an infinite family of MDS codes were obtained.
\item The triply extended code $\overline{\overline{\overline{\C_\alpha}(\bu^{(1)})}(\bu^{(2)})}(-\bone)$ is a $[q+3, q, 3]$ NMDS code over $\gf(q)$ (see Theorem \ref{NNMDS::17}). Hence, another infinite family of NMDS codes were obtained and their weight distributions 
were settled. 
\end{itemize}	
\end{itemize} 
These results above show that multiply extended codes could be very interesting. 
\item The standardly extended codes of several families of cyclic codes were investigated in Section \ref{nsec::5},  and the following main results were obtained:  
\begin{itemize}
\item A family of $[\lambda(\frac{q^m-1}{q-1})+1, \lambda(\frac{q^m-1}{q-1})-m, 3]$ 	codes over $\gf(q)$ were discovered. These codes are dimension-optimal (see Theorem \ref{Cyclic-optimal}). 
\item A family of $[\frac{q^m-1}{e}+1, \frac{q^m-1}{e}-\frac{3m}2, 4]$ codes over $\gf(q)$ were discovered, where $m\geq 2$ is even. These codes are distance-optimal (see Corollary \ref{OPM4:1}).
\item A family of $[q^m+2, q^m+1-2m, 4]$ codes over $\gf(q)$ were discovered, where $q\geq 4$ is even. These codes are distance-optimal (see Corollary \ref{NCOR1}).  
\item A family of $[\frac{q^m+3}2, \frac{q^m+1}2-2m, 4]$ codes over $\gf(q)$ were discovered, where $q^m\equiv 1\pmod{4}$ and $q\geq 11$. These codes are distance-optimal (see Corollary \ref{NCOR2}). 
\item A family of $[\frac{q^m+3}2, \frac{q^m+1}2-2m, 4\leq d\leq 5]$ codes over $\gf(q)$ were discovered, where $q^m\equiv 1\pmod{4}$ and $q\in \{5,7,9\}$. Further, $d=4$ provided that $q^m\equiv -1\pmod{3}$. These codes are distance-almost-optimal with respect to the sphere packing bound (see Corollary \ref{NCOR2}). 
\item A family of $[\frac{3^m+3}2, \frac{3^m+1}2-2m, 5\leq d\leq 6]$ codes over $\gf(3)$ were discovered, where $m\geq 2$ is even. Further, $d=5$ provided that $m\equiv 4\pmod{8}$ or $m\equiv 6\pmod{12}$. These codes are distance-almost-optimal with respect to the sphere packing bound (see Corollary \ref{NCOR2}). 
\item A family of $[\frac{3^m+1}2, \frac{3^m-1}2-2m, 5]$ codes over $\gf(3)$ were discovered, where $m\geq 3$. These codes are distance-almost-optimal with respect to the sphere packing bound (see Corollary \ref{CEXT20}). 
\item A family of $[\frac{3^m-1}{e}+1, \frac{3^m-1}{e}-2m, 5]$ codes over $\gf(3)$ were discovered, where $m\geq 4$ is even. These codes are distance-almost-optimal with respect to the sphere packing bound for $m\geq 1+3\log_3e$ (see Corollary \ref{CC-2}).
\item A family of $[2\cdot 3^s-1, 2\cdot 3^s-2-3s, 5]$ codes over $\gf(3)$ were discovered, where $s\geq 2$. These codes are distance-almost-optimal with respect to the sphere packing bound (see Corollary \ref{CEX25}).
\item A family of $[2\cdot 3^s+3, 2\cdot 3^s+2-4s, 5]$ codes over $\gf(3)$ were discovered, where $s\geq 2$. These codes have good parameters (see Corollary \ref{CEX26}).
\end{itemize}
\item The extended codes of nonbinary Hamming codes $\C$ of length $n=(q^m-1)/(q-1)$ over $\gf(q)$ were investigated in 
Section \ref{Sec:6}, and the following main results were obtained:    
\begin{itemize}
\item For $q=2^s$ with $s \geq 2$ and $m=2$, there exists a family of standardly extended codes $\overline{\cHam(q,m,\ba)}(-\bone)$ of the Hamming codes $\cHam(q,m,\ba)$ with minimum distance $4$ (see Theorem \ref{MDS::1}).	
\item The extended codes $\overline{\cHam(q,m,\ba)}(-\bone)$  of nonbinary Hamming codes have minimum distance $3$, provided that $q$ is odd and $m\geq 2$ is even (see Theorem \ref{thm15}).

\item The extended codes $\overline{\cHam(q,m,\ba)}(-\bone)$  of nonbinary Hamming codes have minimum distance $3$, provided
that $q>2$ and $m\equiv 0\pmod{3}$ (see Theorem \ref{thm16}). 

\item The standardly extended codes of the nonbinary constacyclic Hamming codes have minimum distance $3$ (see Theorem \ref{HAMM}).  
\end{itemize}
\end{enumerate}
In summary,  four infinite families of dimension-optimal or distance-optimal codes were obtained in this paper with the extending technique $\overline{\C}(\bu)$.  In addition, Open Problem \ref{prob-1120} was solved for eleven infinite families of linear codes.

\

{\bf Data availability}~No data was used for the research described in the article.

\end{document}